\documentclass{lmcs}
\pdfoutput=1
\usepackage[utf8]{inputenc}

\usepackage{lastpage}
\lmcsdoi{19}{3}{14}
\lmcsheading{}{\pageref{LastPage}}{}{}%
{Jul.~27,~2022}{Aug.~10,~2023}{}

\pdfoutput=1


\usepackage{enumitem}
\usepackage{amssymb}
\usepackage{tikz-cd}
\usepackage{mathtools}
\usepackage{graphicx}
\usepackage{mathrsfs}
\usepackage{stmaryrd}
\usepackage{bm}
\usepackage{mathrsfs}
\usepackage[scr=euler]{mathalfa}
\usepackage{stackengine}
\usepackage{url}
\usepackage{hyperref}
\usepackage{tabto}

\providecommand{\noopsort}[1]{}

\allowdisplaybreaks

\def\eg{{\em e.g.~}}
\def\cf{{\em cf.~}}
\def\Cf{{\em Cf.~}}

\DeclareMathOperator{\dom}{dom} 
\DeclareMathOperator{\Emb}{\mathbb{S}} 
\DeclareMathOperator{\Path}{\mathbb{P}} 

\DeclareMathOperator{\0}{\mathbf{0}} 
\newcommand{\B}{\mathcal{B}} 
\newcommand{\U}{\mathcal{U}} 
\newcommand{\down}{{\downarrow}} 
\newcommand{\up}{{\uparrow}} 
\newcommand{\pit}{\pitchfork}
\newcommand{\br}[1]{\llbracket #1 \rrbracket} 
\newcommand{\G}{\mathscr{G}} 
\newcommand{\EG}{{}^{\exists^+}\hspace{-2pt}\mathscr{G}} 
\newcommand{\SEG}{{}^{\exists}\mathscr{G}} 
\newcommand{\W}{\mathscr{W}} 
\newcommand{\pro}{\,\bm{\pi}\,} 
\newcommand{\FO}{\mathrm{FO}} 
\newcommand{\EFO}{\exists^+\mathrm{FO}} 

\newcommand{\emb}{\rightarrowtail} 
\newcommand{\lemb}{\leftarrowtail} 
\newcommand{\into}{\hookrightarrow} 
\newcommand{\epi}{\twoheadrightarrow} 
\newcommand{\id}{\mathrm{id}} 
\newcommand{\Q}{\mathscr{Q}} 
\newcommand{\M}{\mathscr{M}} 

\newcommand{\xepi}[2][]{%
  \xrightarrow[#1]{#2}\mathrel{\mkern-14mu}\rightarrow
}

\newcommand{\xemb}[2][]{\ensurestackMath{\mathrel{%
  \stackengine{1pt}{%
    \stackengine{0pt}{\rightarrowtail}{\scriptstyle#2}{O}{c}{F}{F}{S}%
  }{\scriptstyle#1}{U}{c}{F}{F}{S}%
}}}

\newcommand{\xlemb}[2][]{\ensurestackMath{\mathrel{%
  \stackengine{1pt}{%
    \stackengine{0pt}{\leftarrowtail}{\scriptstyle#2}{O}{c}{F}{F}{S}%
  }{\scriptstyle#1}{U}{c}{F}{F}{S}%
}}}

\tikzcdset{
    relay arrow/.default = 10pt,
    relay arrow/.style = {
        rounded corners,
        to path = {
            \pgfextra{
                \def\sourcecoordinate{\pgfpointanchor{\tikztostart}{center}}
                \def\targetcoordinate{\pgfpointanchor{\tikztotarget}{center}}
                \pgfmathanglebetweenpoints{\sourcecoordinate}{\targetcoordinate}
                \edef\tempangle{\pgfmathresult}
                \pgftransformrotate{\tempangle}
                \pgfmathifthenelse{#1>0}{\tempangle+90}{\tempangle-90}
                \pgfcoordinate{tempcoord}{\pgfpointanchor{\tikztostart}{\pgfmathresult}}
            }
            (tempcoord)
            -- ([yshift=#1]tempcoord)
            -- ([yshift=#1]tempcoord-|\tikztotarget.center)\tikztonodes
            --(\tikztotarget)
        }
    }
}

\renewcommand{\epsilon}{\varepsilon}
\renewcommand{\theta}{\vartheta}
\renewcommand{\phi}{\varphi}

\DeclareMathOperator{\C}{\mathscr{C}} 
\DeclareMathOperator{\F}{\mathscr{F}} 
\DeclareMathOperator{\T}{\mathscr{T}} 
\DeclareMathOperator{\E}{\mathscr{E}} 
\newcommand{\CSstar}{\mathsf{Struct}_{\star}(\sg)}

\newcommand{\Fraisse}{Fra\"{i}ss\'{e}}
\newcommand{\ie}{\emph{i.e.}~}
\newcommand{\iec}{\emph{i.e.},~}
\newcommand{\sg}{\sigma}
\newcommand{\RA}{R^{\As}}
\newcommand{\RB}{R^{\Bs}}

\newcommand{\IMP}{\; \Rightarrow \;}

\newcommand{\CS}{\mathsf{Struct}(\sg)}
\newcommand{\struct}[1]{\mathcal{#1}}
\newcommand{\As}{\struct{A}}
\newcommand{\Bs}{\struct{B}}

\DeclareMathOperator{\R}{\mathscr{R}} 
\newcommand{\RT}{\R^{E}}
\newcommand{\RTk}{\R^{E}_{k}}
\newcommand{\RPk}{\R^{P}_{k}}

\newcommand{\RMk}{\R^{M}_{k}}
\newcommand{\cvr}{\prec}

\newcommand{\va}{\vec{a}}
\newcommand{\ve}{\varepsilon}
\newcommand{\Cp}{\C_{p}}
\newcommand{\LE}{L^{E}}
\newcommand{\LP}{L^{P}}
\newcommand{\LM}{L^{M}}
\newcommand{\eqaCk}{\rightleftarrows_{k}^{\C}}
\newcommand{\eqbCk}{\leftrightarrow_{k}^{\C}}
\newcommand{\eqcCk}{\cong_{k}^{\C}}
\newcommand{\eqdCk}{\rightleftharpoons_{k}^{\C}}
\newcommand{\LL}{\mathcal{L}}

\newcommand{\Lk}{\mathcal{L}_k}

\newcommand{\Lck}{\mathcal{L}_{k}(\Count)}
\newcommand{\ELk}{\exists^+\hspace{-1pt}\mathcal{L}_{k}}
\newcommand{\SELk}{\exists\mathcal{L}_{k}}

\newcommand{\eqLk}{\equiv^{\Lk}}
\newcommand{\eqELk}{\equiv^{\ELk}}
\newcommand{\eqSELk}{\equiv^{\SELk}}
\newcommand{\eqLck}{\equiv^{\Lck}}

\newcommand{\eqak}{\rightleftarrows_{k}}
\newcommand{\eqbk}{\leftrightarrow_{k}}
\newcommand{\eqck}{\cong_{k}}
\newcommand{\eqdk}{\rightleftharpoons_{k}}

\newcommand{\eqL}{\equiv^{\LL}}
\newcommand{\Count}{\#}
\newcommand{\vphi}{\varphi}
\newcommand{\IFF}{\; \Longleftrightarrow \;}

\newcommand{\CSplus}{\mathsf{Struct}(\sg^I)}

\begin{document}

\title[Arboreal Categories: An Axiomatic Theory of Resources]{Arboreal Categories:\texorpdfstring{\\}{ }An Axiomatic Theory of Resources}
\titlecomment{{\lsuper*}This is an extended version of the paper \emph{Arboreal Categories and Resources}, which has appeared in the proceedings of the 48th International Colloquium on Automata, Languages, and Programming (ICALP), LIPIcs Vol.\ 198, pp.\ 115:1--115:20, 2021.}

\author[S.~Abramsky]{Samson Abramsky\lmcsorcid{0000-0003-3921-6637}}	
\author[L.~Reggio]{Luca Reggio\lmcsorcid{0000-0001-7331-7381}}	
\address{Department of Computer Science, University College London, United Kingdom}	
\email{l.reggio@ucl.ac.uk, s.abramsky@ucl.ac.uk}  
\thanks{Research supported by the EPSRC project EP/T00696X/1 ``Resources and Co-Resources: a junction between categorical semantics and descriptive complexity'' and by the European Union's Horizon 2020 research and innovation programme under the Marie Sk{\l}odowska-Curie grant agreement No 837724.}	


\begin{abstract}
Game comonads provide a categorical syntax-free approach to finite model theory, and their Eilenberg-Moore coalgebras typically encode important combinatorial parameters of structures. In this paper, we develop a framework whereby the essential properties of these categories of coalgebras are captured in a purely axiomatic fashion.
To this end, we introduce \emph{arboreal categories}, which have an intrinsic process structure, allowing dynamic notions such as bisimulation and back-and-forth games, and resource notions such as number of rounds of a game, to be defined. These are related to extensional or ``static'' structures via \emph{arboreal covers}, which are resource-indexed comonadic adjunctions. These ideas are developed in a general, axiomatic setting, and applied to relational structures, where the comonadic constructions for pebbling, Ehrenfeucht-\Fraisse~and modal bisimulation games recently introduced by Abramsky, Dawar \emph{et al}.\ are recovered, showing that many of the fundamental notions of finite model theory and descriptive complexity arise from instances of arboreal covers.
\end{abstract}

\maketitle

\section{Introduction}

 In previous work (\cite{abramsky2017pebbling,DBLP:conf/csl/AbramskyS18,AS2021}), it has been shown how a range of model comparison games which play a central role in finite model theory, including Ehrenfeucht-\Fraisse, pebbling, and bisimulation games, can be captured in terms of resource-indexed comonads on the category of relational structures and homomorphisms. 
 This was done for $k$-pebble games in~\cite{abramsky2017pebbling}, and extended to 
Ehrenfeucht-\Fraisse~games, and bisimulation games for the modal fragment, in~\cite{DBLP:conf/csl/AbramskyS18}.
In subsequent work, this has been further extended to games for generalized quantifiers \cite{conghaile2021game}, guarded fragments of first-order logic~\cite{Guarded2021}, hybrid and bounded fragments~\cite{AM2022}, and bounded conjunction finite variable logic~\cite{MS2022}. An important feature of this comonadic analysis is that it leads to novel characterisations of important combinatorial parameters such as tree-width and tree-depth. The coalgebras for each of these comonads correspond to certain forms of tree decompositions of structures, with the resource index matching the corresponding combinatorial parameter. 

This leads to the question motivating the present paper:

\vspace{0.5em}
\begin{center}
\fbox{Can we capture the significant common elements of these constructions?}
\end{center}
\vspace{0.5em}

Our aim is to develop an elegant axiomatic account, based on clear conceptual principles, which will yield all these examples and more, and allow a deeper and more general understanding of resources.

Conceptually, a key ingredient is the assignment of a process structure---an intensional description---to an extensional object, such as a function, a set, or a relational structure. It is this process structure, unfolding in space and time, to which a resource parameter can be applied, which can then be transferred to the extensional object. 
At the basic level of computability, this happens when we assign a Turing machine description or a G\"odel number to a recursive function. It is then meaningful to assign  a complexity measure to the function. The same phenomenon arises in semantics: for example, the notion of \emph{sequentiality} is applicable to a \emph{process} computing a higher-order function. Reifying these processes in the form of \emph{game semantics} led to a resolution of the famous full abstraction problem for PCF \cite{abramsky2000full,hyland2000full}, and to a wealth of subsequent results \cite{Church2017}. 

It is now becoming clear that this phenomenon is at play in the game comonads described in \cite{abramsky2017pebbling,DBLP:conf/csl/AbramskyS18,AS2021,conghaile2021game,Guarded2021}. They build tree-structured covers of a given, purely extensional relational structure.  Such a tree cover will in general not have the full properties of the original structure, but be a ``best approximation'' in some resource-restricted setting. More precisely, this means that we have an adjunction, yielding the corresponding comonad. The objects of the category where the approximations live have an intrinsic tree structure, which can be captured axiomatically. The tree encodes a process for generating (parts of) the relational structure, to which resource notions can be applied.

In this paper, we make this intuition precise. We introduce a notion of \emph{arboreal category}, and show how all the examples of game comonads considered to date arise from \emph{arboreal covers}, \ie adjunctions between  extensional categories of relational structures, and arboreal categories. Importantly, these adjunctions are comonadic, and the categories of coalgebras provide a setting for a general notion of bisimulation, which yields a wide range of logical equivalences in the examples. This notion refines the open maps formulation of bisimulation \cite{JNW1993,JNW1996} with the condition that the maps are \emph{pathwise embeddings}, generalizing the ideas introduced in \cite{AS2021}. This allows a much wider range of logical equivalences to be captured.

After recalling some preliminary facts on proper factorisation systems in Section~\ref{s:preliminaries}, we shall develop the axiomatization of paths in Section~\ref{s:paths}. In Section~\ref{s:openmaps} we introduce open pathwise embeddings and bisimulations, and this leads us to the notion of arboreal category. The latter is defined in Section~\ref{s:arboreal}, and in Section~\ref{s:back-and-forth} we establish the correspondence between bisimulations and back-and-forth equivalences in arboreal categories. Finally, in Section~\ref{s:arboreal-covers}, we show how many of the fundamental notions of finite model theory and descriptive complexity arise from instances of arboreal covers.
We shall use the concrete constructions in finite model theory as running examples throughout.

\section{Preliminaries}\label{s:preliminaries}
 
We shall assume familiarity with some standard notions in category theory. All needed background can be found in \cite{adamek2004abstract,MacLane}. All categories under consideration are assumed to be locally small and \emph{well-powered}, \ie every object has a set of subobjects (as opposed to a proper class).

\begin{exa}
The extensional categories of primary interest in this paper are categories of relational structures.
A relational vocabulary $\sg$ is a set of relation symbols $R$, each with a specified positive integer arity.
A $\sg$-structure $\As$ is given by a (possibly empty) set $A$, the universe of the structure, and for each $R$ in $\sg$ with arity $n$, a relation $\RA \subseteq A^n$. A homomorphism $h \colon \As \to \Bs$ is a function $h\colon A \to B$ such that, for each relation symbol $R$ of arity $n$ in $\sg$, for all $a_1, \ldots , a_n$ in $A$,
$\RA(a_1,\ldots , a_n) \IMP \RB(h(a_1), \ldots , h(a_n))$. We write $\CS$ for the category of $\sg$-structures and homomorphisms. 
The Gaifman graph of a structure $\As$ is the graph with vertices $A$, such that two distinct elements are adjacent if they both occur in some tuple $\va \in \RA$ for some relation symbol $R$ in $\sg$.
\end{exa}

To deal with fragments of modal logic, we will consider \emph{(multi-)modal vocabularies}, \ie relational vocabularies in which every relation symbol has arity at most $2$.
If $\sg$ is a modal vocabulary, we can assign to each unary relation symbol $P\in \sg$ a propositional variable $p$, and to each binary relation symbol $R\in \sg$ modalities $\Diamond_R$ and $\Box_R$. In this case, we refer to $\sg$-structures as \emph{Kripke structures}. For any Kripke structure $\As$, the unary relation $P^{\As}$ corresponds to the valuation of the propositional variable $p$, and the binary relation $R^{\As}$ to the accessibility relation for the modalities $\Diamond_R$ and $\Box_R$. 

\subsection{Proper factorisation systems}
We recall the notion of weak factorisation system in a category $\C$.
Given arrows $e$ and $m$ in $\C$, we say that $e$ has the \emph{left lifting property} with respect to $m$, or that $m$ has the \emph{right lifting property} with respect to $e$, if for every commutative square as on the left-hand side below
\begin{equation*}
\begin{tikzcd}
{\bullet} \arrow{d} \arrow{r}{e} & {\bullet} \arrow{d} \\
{\bullet} \arrow{r}{m} & {\bullet}
\end{tikzcd}
\ \ \ \ \ \ \ \ \ \ \ \ \ 
\begin{tikzcd}
{\bullet} \arrow{d} \arrow{r}{e} & {\bullet} \arrow{d} \arrow{dl}[description]{d} \\
{\bullet} \arrow{r}{m} & {\bullet}
\end{tikzcd}
\end{equation*}
there exists a (not necessarily unique) \emph{diagonal filler}, \ie an arrow $d$ such that the right-hand diagram above commutes. If this is the case, we write $e {\,\pit\,} m$. For any class $\mathscr{H}$ of morphisms in $\C$, let ${}^{\pit}\mathscr{H}$ (respectively $\mathscr{H}^{\pit}$) be the class of morphisms having the left (respectively right) lifting property with respect to every morphism in $\mathscr{H}$.

\begin{defi}\label{def:weak-f-s}
A pair of classes of morphisms $(\Q,\M)$ in a category $\C$ is a \emph{weak factorisation system} provided it satisfies the following conditions:
\begin{enumerate}[label=(\roman*)]
\item Every morphism $f$ in $\C$ can be written as $f = m \circ e$ with $e\in \Q$ and $m\in \M$.
\item $\Q={}^{\pit}\M$ and $\M=\Q^{\pit}$.
\end{enumerate}
A \emph{proper factorisation system} is a weak factorisation system $(\Q,\M)$ such that every arrow in $\Q$ is an epimorphism and every arrow in $\M$ a monomorphism. A proper factorisation system is \emph{stable}\footnote{In the literature, the adjective \emph{stable} is usually reserved for the stronger property stating that, for every $e\in \Q$, the pullback of $e$ along any morphism exists and is in~$\Q$.} if, for any $e\in \Q$ and $m\in\M$ with common codomain, the pullback of $e$ along $m$ exists and belongs to $\Q$.
\end{defi}

\begin{rem}
Since every $\Q$-morphism in a proper factorisation system $(\Q,\M)$ is an epimorphism, the diagonal fillers are unique if they exist. That is, any proper factorisation system is \emph{orthogonal}. In particular, factorisations are unique up to (unique) isomorphism.
\end{rem}

\begin{exa}
If $\As$ is a relational structure then, for any $S \subseteq A$, there is an induced substructure with universe $S$.
The inclusion map $S \into A$ is an \emph{embedding}, \ie an injective homomorphism which reflects as well as preserves relations. Any embedding $m\colon \As \to \Bs$ factors as $\As \cong \mathsf{Im}(m) \into \Bs$. Taking $\Q$ to be the surjective homomorphisms and $\M$ to be the embeddings gives a proper factorisation system in $\CS$. This factorisation system is stable because pullbacks in $\CS$ are computed in the category of sets and functions, where (surjections, injections) is a stable proper factorisation system.
\end{exa}
Next, we state some well known properties of weak factorisation systems (\cf \cite{freyd1972categories} or~\cite{riehl2008factorization}):
\begin{lem}\label{l:factorisation-properties}
Let $(\Q,\M)$ be a weak factorisation system in $\C$. The following statements~hold:
\begin{enumerate}[label=(\alph*)]
\item\label{compositions} $\Q$ and $\M$ are closed under compositions.
\item\label{isos} $\Q\cap\M=\{\text{isomorphisms}\}$.
\item\label{pullbacks} The pullback of an $\M$-morphism along any morphism, if it exists, is again in $\M$.
\end{enumerate}
Moreover, if $(\Q,\M)$ is proper, the following hold for all composable morphisms $f,g$ in $\C$:
\begin{enumerate}[label=(\alph*)]\setcounter{enumi}{3}
\item\label{cancellation-e} $g\circ f\in \Q$ implies $g\in\Q$.
\item\label{cancellation-m} $g\circ f\in\M$ implies $f\in\M$.\qed
\end{enumerate}
\end{lem}

Throughout this paper, we will refer to $\M$-morphisms as \emph{embeddings} and denote them by $\emb$. $\Q$-morphisms will be referred to as \emph{quotients} and denoted by $\epi$. 

Assume $\C$ is a category admitting a proper factorisation system $(\Q,\M)$. In the same way that one usually defines the poset of subobjects of a given object $X\in\C$, we can define the poset of $\M$-subobjects of $X$. Given embeddings $m\colon S\emb X$ and $n\colon T\emb X$, let us say that $m\trianglelefteq n$ provided there is a morphism $i\colon S\to T$ such that $m=n\circ i$.
\[\begin{tikzcd}
S \arrow[rightarrowtail]{r}{m} \arrow[dashed, swap]{d}{i} & X \\
T \arrow[rightarrowtail]{ur}[swap]{n} & {}
\end{tikzcd}\] 
(Note that, if it exists, $i$ is necessarily an embedding by Lemma~\ref{l:factorisation-properties}\ref{cancellation-m}.)
This yields a preorder on the class of all embeddings with codomain $X$. The symmetrization\label{symmetriz-preorder-subobjects} $\sim$ of $\trianglelefteq$ can be characterised as follows: $m\sim n$ if, and only if, there exists an isomorphism $i\colon S\to T$ such that $m=n\circ i$. Let $\Emb{X}$ be the class of $\sim$-equivalence classes of embeddings with codomain $X$, equipped with the natural partial order $\leq$ induced by~$\trianglelefteq$. We shall systematically represent a $\sim$-equivalence class by any of its representatives. As $\C$ is well-powered and every embedding is a monomorphism, we see that $\Emb{X}$ is a set. 

For any morphism $f\colon X\to Y$ and embedding $m\colon S\emb X$, we can consider the $(\Q,\M)$-factorisation $S\epi \exists_f S \emb Y$ of $f\circ m$. This yields a monotone map $\exists_f\colon \Emb{X}\to\Emb{Y}$ sending $m$ to the embedding $\exists_f S\emb Y$. (Note that the map $\exists_f$ is well-defined because factorisations are unique up to isomorphism.) Further, if $(\Q,\M)$ is stable and $f$ is a quotient, we let $f^*\colon \Emb{Y}\to \Emb{X}$ be the monotone map sending $n\colon T\emb Y$ to its pullback along $f$. 

\begin{lem}\label{l:adj-image-pullback}
Let $\C$ be a category equipped with a stable proper factorisation system, and let $f\colon X\epi Y$ be a quotient in $\C$. The following is an adjoint pair of monotone maps:
\[ \begin{tikzcd}
\Emb{X} \arrow[r, bend left=25, ""{name=U, below}, "\exists_f"{above}]
\arrow[r, leftarrow, bend right=25, ""{name=D}, "f^*"{below}]
& \Emb{Y}
\arrow[phantom, "\textnormal{\footnotesize{$\bot$}}", from=U, to=D] 
\end{tikzcd}
\]
\end{lem}
\begin{proof}
This is a well known fact for posets of (strong) subobjects, see \eg \cite[Proposition~4.4.6]{Borceux1} or \cite[Proposition~3 p.~186]{MM1994}, and the same reasoning applies in our situation. For the sake of completeness, we provide a detailed proof.

Fix arbitrary embeddings $m\colon S\emb X$ and $n\colon T\emb Y$. We must prove that
\[
\exists_f m \leq n \ \Longleftrightarrow \ m\leq f^*n.
\]
Consider the (quotient, embedding) factorisation of $f\circ m$:
\[\begin{tikzcd}
S \arrow[twoheadrightarrow]{r}{e} & \exists_f S \arrow[rightarrowtail]{r}{\exists_f m} & X
\end{tikzcd}\]
Note that $\exists_f m \leq n$ if, and only if, there is $j\colon S\to T$ such that $f\circ m = n\circ j$. Just observe that if such a $j$ exists then the following square commutes, and so it admits a diagonal filler. In particular, the commutativity of the lower triangle implies that $\exists_f m \leq n$.
\[\begin{tikzcd}
S \arrow[twoheadrightarrow]{r}{e} \arrow{d}[swap]{j} & \exists_f S \arrow[rightarrowtail]{d}{\exists_f m} \arrow[dashed]{dl} \\
T \arrow[rightarrowtail]{r}{n} & Y
\end{tikzcd}\]
Conversely, if $\exists_f m \leq n$ there is $i\colon \exists_f S\to T$ such that $n\circ i = \exists_f m$ and we take $j\coloneqq i\circ e$.

Next, consider the following diagram:
\[\begin{tikzcd}[column sep=3em, row sep=2.5em]
S \arrow[bend left=30,dashed]{drr}{j} \arrow[bend right=30,rightarrowtail]{ddr}[swap]{m} \arrow[dashed]{dr}[description]{\xi} & & \\
& f^* T \arrow[rightarrowtail]{d}[swap]{f^* n} \arrow[twoheadrightarrow]{r}{g} \arrow[dr, phantom, "\lrcorner", very near start] & T \arrow[rightarrowtail]{d}{n} \\
 & X \arrow[twoheadrightarrow]{r}{f} & Y
\end{tikzcd}\]
If $\exists_f m \leq n$, there is a morphism $j$ making the outer diagram commute and so, by the universal property of pullbacks, there is a (unique) mediating morphism $\xi$. In particular, the commutativity of the leftmost triangle entails that $m\leq f^*n$. In the other direction, if $m\leq f^*n$ there is $\xi$ such that $f^*n\circ \xi = m$ and thus the morphism $j\coloneqq g\circ \xi$ satisfies $f\circ m = n\circ j$. Therefore, $\exists_f m \leq n$.
\end{proof}

\begin{lem}\label{l:emb-quo-order-embeddings}
Let $\C$ be a category equipped with a stable proper factorisation system, and let $f\colon X\to Y$ be any morphism in $\C$. The following statements hold:
\begin{enumerate}[label=(\alph*)]
\item\label{exists-embed} If $f$ is an embedding, then $\exists_f \colon \Emb{X}\to \Emb{Y}$ is an order-embedding.
\item\label{pullback-embed} If $f$ is a quotient, then $f^*\colon \Emb{Y}\to \Emb{X}$ is an order-embedding.
\end{enumerate}
\end{lem}
\begin{proof}
For item~\ref{exists-embed} note that, as $f\colon X\to Y$ is an embedding, $\exists_f \colon \Emb{X}\to \Emb{Y}$ sends $m$ to $f\circ m$. Let $m_1\colon S_1\emb X$ and $m_2\colon S_2\emb X$ be embeddings such that $f\circ m_1\leq f\circ m_2$. Then there exists $k\colon S_1\to S_2$ such that $f\circ m_1=f\circ m_2\circ k$. Because $f$ is a monomorphism, it follows that $m_1=m_2\circ k$, \ie $m_1\leq m_2$. Hence, $\exists_f$ is an order-embedding.
 
For item~\ref{pullback-embed}, it is enough to prove that $\exists_f f^* n=n$ for any $n\colon T\emb Y$, for then $f^* n_1\leq f^* n_2$ implies $n_1= \exists_f f^* n_1\leq \exists_f f^* n_2=n_2$. Consider the pullback of $f$ along $n$, as displayed on the left-hand side below. 
\begin{center}
\begin{tikzcd}[column sep=3em, row sep=2.5em]
f^* T \arrow[rightarrowtail]{d}[swap]{f^* n} \arrow[twoheadrightarrow]{r} \arrow[dr, phantom, "\lrcorner", very near start] & T \arrow[rightarrowtail]{d}{n} \\
 X \arrow[twoheadrightarrow]{r}{f} & Y
\end{tikzcd}
\ \ \ \ \ \ \ \ \ 
\begin{tikzcd}[column sep=3em, row sep=2.5em]
f^* T \arrow[twoheadrightarrow]{d} \arrow[twoheadrightarrow]{r} & T \arrow[rightarrowtail]{d}{n} \arrow[dashed]{dl} \\
 \exists_f f^* T \arrow[rightarrowtail]{r}{\exists_f f^* n} & Y
\end{tikzcd}
\end{center}
Since the square on the right-hand side above commutes, there exists a diagonal filler $T\to \exists_f f^*T$. Note that this diagonal filler must be both a quotient and an embedding, hence an isomorphism. Therefore, $\exists_f f^* n=n$ in $\Emb{Y}$.
\end{proof}


\section{Path Categories}\label{s:paths}
Throughout this section, we fix a category $\C$ equipped with a stable proper factorisation system.

\subsection{Paths}
If $(P, {\leq})$ is a poset, then a subset $C \subseteq P$ is a \emph{chain}  if it is linearly ordered. $(P,\leq)$ is a \emph{forest} if, for all $x\in P$, the set $\down x\coloneqq \{y\in P\mid y\leq x\}$ is a finite chain. 
The \emph{height} of a forest is the supremum of the cardinalities of its chains (if it exists, otherwise we say that the height is $\infty$). If $(P,\leq)$ is a forest, the \emph{height} of an element $x\in P$ is the cardinality of the set $\down x \setminus \{x\}$.
The \emph{covering relation} $\cvr$ associated with a partial order $\leq$ is defined by $u\cvr v$ if and only if $u<v$ and there is no $w$ such that $u<w< v$. It is convenient to allow the empty forest.
The \emph{roots} of a forest are the minimal elements. A \emph{tree} is a forest with at most one root.\footnote{Note that a tree is either empty, or has a unique root, the least element in the order.}
Morphisms of forests are maps that preserve roots and the covering relation; equivalently, monotone maps that preserve the height of elements.
The category of forests is denoted by~$\F$, and the full subcategory of non-empty trees by~$\T$. We equip these categories with the factorisation system (surjective morphisms, injective morphisms).

\begin{rem}
The factorisation system (surjective morphisms, injective morphisms) in $\F$ is obviously proper. Further, this factorisation system is stable and pullbacks of quotients along embeddings are computed in the category of sets and functions. To see this, consider forest morphisms $f\colon F\epi H$ and $g\colon G\emb H$. Let 
\[
E\coloneqq \{(x,y)\in F\times G\mid f(x)=g(y)\},
\] 
and equip $E$ with the componentwise order. As any injective forest morphism is an order-embedding, for all $(x,y),(x',y')\in E$ we have $(x,y)\leq (x',y')$ if and only if $x\leq x'$. Just observe that $x\leq x'$ entails
\[
g(y)=f(x)\leq f(x')=g(y')
\]
and so, since $g$ is an injective forest morphism, $y\leq y'$.
Thus, $E$ is a forest and the height of an element $(x,y)\in E$ coincides with the height of $x$ in $F$, and thus also with the height of $y$ in $G$. The following diagram is then readily seen to be a pullback square in~$\F$,
\[\begin{tikzcd}
E \arrow[twoheadrightarrow]{r}{\pi_G} \arrow[rightarrowtail]{d}[swap]{\pi_F} & G \arrow[rightarrowtail]{d}{g} \\
F \arrow[twoheadrightarrow]{r}{f} & H
\end{tikzcd}\]
where $\pi_F$ and $\pi_G$ are the obvious projections.

This factorisation system restricts to a stable proper factorisation system in the category $\T$ of non-empty trees.
\end{rem}
\begin{defi}
An object $X$ of $\C$ is called a \emph{path} provided the poset $\Emb{X}$ is a finite chain. Paths will be denoted by $P,Q,R,\ldots$.
\end{defi}

\begin{exa}\label{Forestex}
The paths in $\F$ are the finite chains, \ie the trees consisting of a single branch. Similarly, in $\T$ the paths are the non-empty finite chains.
\end{exa}

\begin{exas}\label{RTex}
We define a \emph{forest-ordered $\sg$-structure} $(\As, {\leq})$ to be a $\sg$-structure $\As$ with a forest order $\leq$ on $A$.
A morphism of forest-ordered $\sg$-structures $f\colon (\As, {\leq}) \to (\Bs, {\leq'})$ is a $\sg$-homomorphism $f\colon \As \to \Bs$ that is also a forest morphism.  This determines a category $\R(\sg)$. 
We equip $\R(\sg)$ with the factorisation system given by (surjective morphisms, embeddings), where an embedding is a morphism which is an embedding \emph{qua} $\sg$-homomorphism.

In \cite{AS2021}, it is shown that the categories of coalgebras for the various comonads studied there are given, up to isomorphism, by subcategories of $\R(\sg)$ (or minor variants thereof):
\begin{itemize}
    \item For the Ehrenfeucht-\Fraisse~comonad, this is the full subcategory $\RT(\sg)$ determined by those objects satisfying
the condition (E): adjacent elements of the Gaifman graph of $\As$ are comparable in the forest order.
For each $k>0$, $\RTk(\sg)$ is the full subcategory of $\RT(\sg)$ of those forest orders of height $\leq k$.
The objects $(\As, {\leq})$ of $\RTk(\sg)$ are forest covers of $\As$ witnessing that its tree-depth is $\leq k$ \cite{nevsetvril2006tree}.
\item For the pebbling comonad, for each $k>0$ this is the category $\RPk(\sg)$ whose objects have the form $(\As, {\leq}, p)$, where $(\As, {\leq})$ is a forest-ordered $\sg$-structure, and $p\colon A \to \{1,\ldots,k\}$ is a pebbling function. In addition to condition (E), these structures have to satisfy the condition (P): if $a$ is adjacent to $b$ in the Gaifman graph of $\As$, and $a < b$ in the forest order, then for all $x$ such that $a < x \leq b$, $p(a) \neq p(x)$.
It is shown in \cite{AS2021} that these structures are equivalent to the more familiar form of tree decomposition used to define tree-width \cite{kloks1994treewidth}. Morphisms have to preserve the pebbling function.
\item Suppose $\sg$ is a modal vocabulary. For the modal comonad, the category $\RMk(\sg)$ has as objects the non-empty tree-ordered Kripke structures $\As$ of height $\leq k$ satisfying the condition (M): for $x, y \in A$, $x \cvr y$ if and only if for some unique binary relation $R$ in $\sg$ (``transition relation''), $R^{\As}(x, y)$.
\end{itemize}

The paths in each of these categories are those structures  in which the order is a finite chain. These are our key motivating examples for paths. Note that in the modal case, ignoring the interpretations of the propositional variables (\iec the unary relations), these correspond to synchronization trees consisting of a single branch, \ie traces.
\end{exas}

The following fact is an immediate consequence of Lemma~\ref{l:emb-quo-order-embeddings}:
\begin{lem}\label{path-images-and-subs}
Let $f\colon X\to Y$ be any morphism in $\C$. The following statements hold:
\begin{enumerate}[label=(\alph*)]
\item\label{sub-path} If $Y$ is a path and $f$ is an embedding, then $X$ is a path.
\item\label{image-path} If $X$ is path and $f$ is a quotient, then $Y$ is a path.\qed
\end{enumerate}
\end{lem}

\begin{rem}
Note that Lemma~\ref{l:emb-quo-order-embeddings}\ref{exists-embed}, and thus also Lemma~\ref{path-images-and-subs}\ref{sub-path}, holds more generally for any category equipped with a (possibly non-stable) proper factorisation system.
\end{rem}

A \emph{path embedding} is an embedding $P\emb X$ whose domain is a path. 
Given any object $X$ of $\C$, we let $\Path{X}$ be the sub-poset of $\Emb{X}$ consisting of the path embeddings. By Lemma~\ref{path-images-and-subs}\ref{image-path}, for any arrow $f\colon X\to Y$, the monotone map $\exists_f\colon \Emb{X}\to\Emb{Y}$ restricts to a monotone map
\[
\Path{f}\colon \Path{X}\to\Path{Y}, \ \ (m\colon P\emb X)\mapsto (\exists_f m\colon \exists_f P\emb Y).
\]
By the uniqueness up to isomorphism of factorisations, this assignment is functorial.

\subsection{Path categories}
We now define the notion of path category. Suppose for a moment that $\C$ admits all coproducts of small families of paths.
An object $X$ of $\C$ is said to be \emph{connected} if, for all non-empty small families of paths $\{P_i\mid i\in I\}$ in~$\C$, any morphism $X\to \coprod_{i\in I}{P_i}$ factors through some coproduct injection $P_j\to \coprod_{i\in I}{P_i}$.

\begin{rem}
In category theory, an object $a$ of a locally small category $\mathscr{A}$ is usually said to be connected if the functor $\hom_{\mathscr{A}}(a,-)$, from $\mathscr{A}$ into the category of sets and functions, preserves all coproducts that exist in $\mathscr{A}$. The definition that we have given above differs in two aspects: firstly we only consider (non-empty) coproducts of paths, and secondly we do not require that a morphism $X\to \coprod_{i\in I}{P_i}$ factor through a \emph{unique} coproduct injection. Note that, unlike with the traditional notion, with our definition the initial object is connected.
\end{rem}

\begin{defi}\label{def:path-cat}
A \emph{path category} is a category $\C$, equipped with a stable proper factorisation system, that satisfies the following conditions:
\begin{enumerate}[label=(\roman*)]
\item\label{ax:colimits} $\C$ has all coproducts of small families of paths.
\item\label{ax:connected} Every path in $\C$ is connected.
\item\label{ax:2-out-of-3} For any paths $P,Q,R$, if a composite $P\to Q \to R$ is a quotient, then so is $P\to Q$.
\end{enumerate}
\end{defi}

\begin{rem}
Item~\ref{ax:2-out-of-3} in the definition of path category is equivalent to the following condition: For any paths $P,Q,R$ and morphisms 
\begin{equation}\label{eq:2-of-3}
\begin{tikzcd}
P \arrow{r}{f} & Q \arrow{r}{g} & R,
\end{tikzcd}
\end{equation}
if any two of $f$, $g$, and $g\circ f$ are quotients, then so is the third. Thus, we shall refer to~\ref{ax:2-out-of-3} as the \emph{2-out-of-3 condition}. 

To see that the two conditions are equivalent, consider the diagram in equation~\eqref{eq:2-of-3} and suppose that item~\ref{ax:2-out-of-3} above holds. If $f,g$ are quotients, then so is $g\circ f$ by Lemma~\ref{l:factorisation-properties}\ref{compositions}, and if $g\circ f$ is a quotient then $g$ is a quotient by Lemma~\ref{l:factorisation-properties}\ref{cancellation-e} and $f$ is a quotient by item~\ref{ax:2-out-of-3}. Conversely, an application of Lemma~\ref{l:factorisation-properties}\ref{cancellation-e} show that item~\ref{ax:2-out-of-3} holds if, whenever any two of $f$, $g$, and $g\circ f$ are quotients, then so is the third.
\end{rem}

Note that any path category has an initial object $\0$, obtained as the coproduct of the empty family.

\begin{exas}
$\F$ and $\T$ are path categories.
Coproducts of forests are given by disjoint union; in particular, the initial object of $\F$ is the empty forest. For trees, non-empty coproducts are given by smash sum, with the roots identified, and the initial object of $\T$ is the one-element tree. Finite chains are connected in both $\F$ and $\T$. Finally, since forest morphisms preserve height, we see that $\F$ and $\T$ satisfy the 2-out-of-3 condition. A similar reasoning shows that $\R(\sg)$ and its subcategories mentioned in Example~\ref{RTex} are all path categories. Note that coproducts in $\R(\sg)$ are given by disjoint unions. In particular, $\R(\sg)$ has an initial object because we allow empty $\sg$-structures.
\end{exas}

\begin{thm}\label{thm:path-cat-functor-into-trees}
Let $\C$ be a path category.
The assignment $X\mapsto \Path{X}$ induces a functor $\Path\colon \C\to\T$ into the category of non-empty trees.
\end{thm}

To prove this theorem, we start by showing that each poset $\Path{X}$ is an object of $\T$.

\begin{lem}\label{l:PathX-tree}
Let $\C$ be a path category. For any object $X$ of $\C$, $\Path{X}$ is a non-empty tree.
\end{lem}
\begin{proof}
By Lemmas~\ref{l:emb-quo-order-embeddings}\ref{exists-embed} and~\ref{path-images-and-subs}\ref{sub-path}, the sub-poset of $\Path{X}$ consisting of those elements that are below a given $P\in\Path{X}$ is isomorphic to $\Emb{P}$. Since the latter is a finite chain, $\Path{X}$ is a forest. Now, let $\0\epi \widetilde{\0} \xemb{m} X$ be the (quotient, embedding) factorisation of the unique morphism from the initial object to $X$. We claim that $m\colon \widetilde{\0}\emb X$ is the unique root of~$\Path{X}$. 

Note that $\0$ is a path: just observe that any embedding $S\emb \0$ admits the unique morphism $\0\to S$ as a right inverse, and thus is a retraction. It follows that $S\cong \0$, \iec $\Emb{\0}$ is the one-element poset. In particular, $\0$ is a path. 
Thus, $\widetilde{\0}$ is a path by Lemma~\ref{path-images-and-subs}\ref{image-path}. 
We show that $m\colon \widetilde{\0}\emb X$ is the least element of $\Path{X}$. 
If $m'\colon P\emb X$ is any path embedding, we have a commutative square as follows.
\[\begin{tikzcd}
\0 \arrow[twoheadrightarrow]{r} \arrow{d} & \widetilde{\0} \arrow[rightarrowtail]{d}{m} \\
P \arrow[rightarrowtail]{r}{m'} & X
\end{tikzcd}\]
Hence there exists a diagonal filler $d\colon \widetilde{\0}\to P$, and so $m\leq m'$ in $\Path{X}$.
\end{proof}

We next show that the functor $\Path$ sends morphisms in a path category to tree morphisms, thus completing the proof of Theorem~\ref{thm:path-cat-functor-into-trees}.
\begin{prop}\label{p:Path-f-tree-morphism}
Let $\C$ be a path category. For any arrow $f$ in $\C$, $\Path{f}$ is a tree morphism.
\end{prop}
\begin{proof}
Since $\Path{f}$ is monotone, it suffices to prove that it preserves the height of elements. In turn, this is equivalent to saying that, for any path embedding $m\colon P\emb X$, the induced map $\Path{f}\colon \down m\to \down \Path{f}(m)$ is a bijection.

We start by establishing surjectivity, \ie $\down\Path{f}(m)\subseteq \Path{f}(\down m)$. Let $(e,j)$ be the (quotient, embedding) factorisation of $f\circ m$. If $n\colon Q\emb Y$ is a path embedding such that $n\leq \Path{f}(m)$ in $\Path{Y}$, there exists an embedding $k\colon Q\emb \exists_f P$ such that the left-hand diagram below commutes. Consider the pullback of $k$ along $e$, as displayed in the right-hand diagram below.
\begin{center}
\begin{tikzcd}
P \arrow[rightarrowtail]{r}{m} \arrow[twoheadrightarrow]{dr}[swap]{e} & X \arrow{r}{f} & Y & \\
{} & \exists_f P \arrow[rightarrowtail]{ur}[description]{j} & & Q \arrow[rightarrowtail]{ul}[swap]{n} \arrow[rightarrowtail]{ll}[swap]{k}
\end{tikzcd}
\ \ \ \ \ \ \ \ \ \ \ 
\begin{tikzcd}
R \arrow[twoheadrightarrow]{r} \arrow[rightarrowtail]{d} \arrow[dr, phantom, "\lrcorner", very near start] & Q \arrow[rightarrowtail]{d}{k} \\
P \arrow[twoheadrightarrow]{r}{e} & \exists_f P
\end{tikzcd}
\end{center}
Then $R$ is a path by Lemmas~\ref{l:factorisation-properties}\ref{pullbacks} and~\ref{path-images-and-subs}\ref{sub-path}, and the composite $i\colon R\emb P\emb X$ is a path embedding which is below $m$ in the poset $\Path{P}$. Further, the top horizontal arrow in the pullback square is a quotient and so, by the uniqueness up to isomorphism of factorisations, the (quotient, embedding) factorisation of $f\circ i$ is $R\epi Q \xemb{n} Y$, \iec $\Path{f}(i)=n$.

For injectivity, let $m_1\colon P_1\emb X$ and $m_2\colon P_2\emb X$ be path embeddings in $\down m$. Since $P$ is a path, $m_1$ and $m_2$ are comparable in the order of $\Path{X}$. Assume without loss of generality that $m_1\leq m_2$, \iec there exists an embedding $k\colon P_1\emb P_2$ such that $m_1=m_2\circ k$. 
If $\Path{f}(m_1)=\Path{f}(m_2)$, there exists an isomorphism $k'\colon \exists_f P_1\to \exists_f P_2$ making the left-hand diagram below commute. 
\begin{center}
\begin{tikzcd}
P_1 \arrow[twoheadrightarrow]{r} \arrow[rightarrowtail]{d}[swap]{m_1} \arrow[dd, relay arrow=-3ex, "k", swap] & \exists_f P_1 \arrow[rightarrowtail]{d} \arrow[dd, relay arrow=3ex, "k'"] \\
X \arrow{r}{f} & Y \\
P_2 \arrow[twoheadrightarrow]{r} \arrow[rightarrowtail]{u}{m_2} & \exists_f P_2 \arrow[rightarrowtail]{u}
\end{tikzcd}
\ \ \ \ \ \ \ \ \ \ 
\begin{tikzcd}
P_1 \arrow[rightarrowtail]{r}{k} \arrow[rr, relay arrow=2ex, "", twoheadrightarrow] & P_2 \arrow[twoheadrightarrow]{r} & \exists_f P_2
\end{tikzcd}
\end{center}
In particular, the diagram on the right-hand side above commutes, where the top horizontal arrow is the composition of $P_1\epi \exists_f P_1$ with the isomorphism $k'\colon \exists_f P_1 \to \exists_f P_2$. By the 2-out-of-3 condition, $k$ is an isomorphism, and so $m_1=m_2$ in $\Path{X}$.
\end{proof}

\begin{rem}
Direct inspection shows that Theorem~\ref{thm:path-cat-functor-into-trees} holds, more generally, for any category $\C$, equipped with a stable proper factorisation system, that satisfies condition~\ref{ax:2-out-of-3} in Definition~\ref{def:path-cat} and admits an initial object. In the absence of an initial object, the proof of the aforementioned theorem shows that $\Path$ yields a functor into the category $\F$ of forests.
\end{rem}

We conclude this section with the following observation, which will play an important role in the context of arboreal categories (\cf Lemma~\ref{l:arboreal-consequences}, and Propositions~\ref{pr:perfect-lattice-of-strong-subs} and~\ref{p:b-and-f-implies-bisimilar}):

\begin{lem}\label{l:path-cat-suprema}
The following statements hold for any object $X$ of a path category $\C$:
\begin{enumerate}[label=(\alph*)]
\item\label{suprema-of-paths} Any subset $\U\subseteq \Path{X}$ admits a supremum $\bigvee \U$ in $\Emb{X}$.
\item\label{paths-j-irred} For any path embedding $m\in\Path{X}$ and non-empty set $\mathcal{S}\subseteq \Emb{X}$, if $m= \bigvee\mathcal{S}$ then $m\in\mathcal{S}$.
\item\label{at-most-one-emb} Between any two paths there is at most one embedding.
\end{enumerate} 
\end{lem}
\begin{proof}
For item~\ref{suprema-of-paths}, consider a set of path embeddings 
\[
\U=\{m_i\colon P_i\emb X\mid i\in I\}\subseteq \Path{X}.
\] 
Let $S\coloneqq \coprod_{i\in I} P_i$ be the coproduct in $\C$ of the paths $P_i$ and consider the (quotient, embedding) factorisation of the unique morphism $\delta\colon S\to X$ whose component at $P_i$ is $m_i$:
\[\begin{tikzcd}
S \arrow[twoheadrightarrow]{r}{e} \arrow[rr, relay arrow=2ex, "\delta"] & T \arrow[rightarrowtail]{r}{m} & X
\end{tikzcd}\] 
Each path embedding $m_i\in \U$ factors through $m$, thus $m$ is an upper bound for $\U$. We claim that $m$ is the least upper bound, \iec $m=\bigvee{\U}$ in $\Emb{X}$. Suppose that all path embeddings in $\U$ factor through some embedding $m'\colon T'\emb X$. By the universal property of $S$, we get a morphism $\phi\colon S\to T'$. By uniqueness of $\delta$ we obtain $m'\circ \phi=\delta$, and so the following square commutes.
\[\begin{tikzcd}
S \arrow{d}[swap]{\phi} \arrow[twoheadrightarrow]{r}{e} & T \arrow[rightarrowtail]{d}{m} \\
T' \arrow[rightarrowtail]{r}{m'} & X
\end{tikzcd}\]
Therefore, there exists a diagonal filler $T\to T'$. In particular, the commutativity of the lower triangle entails that $m\leq m'$, as was to be proved.

For item~\ref{paths-j-irred}, let $m\colon P\emb X$ be a path embedding and let $\mathcal{S}\subseteq \Emb{X}$ be a non-empty set such that $m=\bigvee{\mathcal{S}}$. Then $n\leq m$ for each $n\in\mathcal{S}$. Since $P$ is a path, $\down m$ is a finite chain in $\Emb{X}$ and so $\mathcal{S}$ must be a finite set whose largest element is $m$. In particular, $m\in\mathcal{S}$.

Finally, for item~\ref{at-most-one-emb}, since there is at most one tree morphism between any two finite chains, it suffices to show that $\Path\colon\C\to\T$ is faithful on embeddings between paths. That is, whenever $m,n\colon P\emb Q$ are embeddings between paths, $\Path{m}=\Path{n}$ implies $m=n$. In turn, this follows at once by evaluating $\Path{m}$ and $\Path{n}$ at $\id_P$.
\end{proof}

\begin{rem}
All results above concerning path categories remain valid if we drop item~\ref{ax:connected} in Definition~\ref{def:path-cat}, stating that every path is connected. For this reason, this condition was not included in the original definition of path category in \cite{AR2021icalp}. We do include it here as this allows us to give a conceptually simpler definition of arboreal category, \cf Definition~\ref{def:arboreal-cat}.
\end{rem}

\section{Pathwise Embeddings, Open Maps, and Bisimulations}\label{s:openmaps}
Throughout this section, we fix a category $\C$ equipped with a stable proper factorisation system.

\subsection{Pathwise embeddings and open maps}
Following \cite{AS2021}, let us say that a morphism $f\colon X\to Y$ in $\C$ is a \emph{pathwise embedding} if, for all path embeddings $m\colon P\emb X$, the composite $f\circ m$ is a path embedding. Hence, $\Path{f}(m)=f\circ m$ for all $m\in\Path{X}$. 
Following again \cite{AS2021}, we introduce a notion of open map---inspired by~\cite{JNW1993}---that, combined with the concept of pathwise embedding, will allow us to define an appropriate notion of bisimulation.
A morphism $f\colon X\to Y$ in $\C$ is said to be \emph{open} if it satisfies the following path-lifting property: Given any commutative square
\[\begin{tikzcd}
P \arrow[rightarrowtail]{r} \arrow[rightarrowtail]{d} & Q \arrow[rightarrowtail]{d} \arrow[dashed]{dl} \\
X \arrow{r}{f} & Y
\end{tikzcd}\]
with $P,Q$ paths, there exists a diagonal filler $Q\to X$, \ie an arrow $Q\to X$ making the two triangles commute. Note that, if it exists, such a diagonal filler is an embedding because so is the right vertical arrow (see Lemma~\ref{l:factorisation-properties}\ref{cancellation-m}).

The previous definition of open map differs from the one given in~\cite{JNW1993} because we require that, in the square above, the top horizontal morphism and the vertical ones be embeddings. 
However, the next lemma shows that a morphism is an open pathwise embedding if, and only if, it is open in the sense of~\cite{JNW1993} (with respect to the full subcategory of $\C$ defined by the paths). \Cf also Remark~\ref{rem:arboreal-open-vs-JNW} below.

\begin{lem}\label{l:open-pwemb-JNW-open}
The following statements are equivalent for any morphism $f\colon X\to Y$:
\begin{enumerate}[label=(\arabic*)]
\item $f$ is an open pathwise embedding.
\item Given any commutative square
\[\begin{tikzcd}
P \arrow{r} \arrow{d} & Q \arrow{d} \arrow[dashed]{dl} \\
X \arrow{r}{f} & Y
\end{tikzcd}\]
with $P,Q$ paths, there exists a diagonal filler $Q\to X$.
\end{enumerate}
\end{lem}
\begin{proof}
$(1)\Rightarrow (2)$. Consider a commutative square
\begin{equation}\label{eq:strong-open-square}
\begin{tikzcd}
P \arrow{r}{h} \arrow{d}[swap]{g_1} & Q \arrow{d}{g_2}  \\
X \arrow{r}{f} & Y
\end{tikzcd}
\end{equation}
where $P$ and $Q$ are paths. Let $(e_1,m_1)$ and $(e_2,m_2)$ be the (quotient, embedding) factorisations of $g_1$ and $g_2$, respectively. Because $e_1$ has the left lifting property with respect to $m_2$, the diagram on the left-hand side below admits a diagonal filler $d\colon P'\to Q'$.
\begin{center}
\begin{tikzcd}[row sep=3em]
P \arrow[twoheadrightarrow]{d}[swap]{e_1} \arrow{r}{h} & Q \arrow[twoheadrightarrow]{r}{e_2} & Q' \arrow[rightarrowtail]{d}{m_2} \\
P' \arrow[rightarrowtail]{r}{m_1} \arrow[dashed]{urr}[description]{d} & X \arrow{r}{f} & Y
\end{tikzcd}
\ \ \ \ \ \ \ \ \ 
\begin{tikzcd}
P' \arrow[rightarrowtail]{r}{d} \arrow[rightarrowtail]{d}[swap]{m_1} & Q' \arrow[rightarrowtail]{d}{m_2} \arrow[dashed]{dl}[description]{k} \\
X \arrow{r}{f} & Y
\end{tikzcd}
\end{center}
Note that $P'$ and $Q'$ are paths by Lemma~\ref{path-images-and-subs}\ref{image-path}. Further, $f\circ m_1$ is an embedding because $f$ is a pathwise embedding, and so $d$ is an embedding by Lemma~\ref{l:factorisation-properties}\ref{cancellation-m}. As $f$ is open, the square on the right-hand side above admits a diagonal filler $k\colon Q'\to X$. The composite morphism $k\circ e_2\colon Q \to X$ is a diagonal filler of diagram~\eqref{eq:strong-open-square}. Just observe that
\[
(k\circ e_2) \circ h = k \circ d\circ e_1 = m_1 \circ e_1 = g_1
\]
and 
\[
f\circ (k\circ e_2) = m_2 \circ e_2 = g_2.
\]

$(2)\Rightarrow (1)$. It suffices to show that $f$ is a pathwise embedding. Given a path embedding $m\colon P \emb X$, take the (quotient, embedding) factorisation of $f\circ m \colon P\to Y$:
\[\begin{tikzcd}
P \arrow[twoheadrightarrow]{r}{h} \arrow[rightarrowtail]{d}[swap]{m} & Q \arrow[rightarrowtail]{d}{n} \\
X \arrow{r}{f} & Y
\end{tikzcd}\]
Note that $Q$ is a path by Lemma~\ref{path-images-and-subs}\ref{image-path}, and so the square admits a diagonal filler $d\colon Q \to X$. The identity $d\circ h = m$ implies that $h$ is an embedding, hence the composite $n\circ h = f\circ m$ is also an embedding. That is, $f$ is a pathwise embedding.
\end{proof}

For a pathwise embedding $f\colon X\to Y$, openness can be characterised in terms of the corresponding morphism of trees $\Path{f}\colon \Path{X}\to \Path{Y}$:
\begin{prop}\label{p:pwe-open-dual-p-morph}
The following statements are equivalent for any pathwise embedding $f\colon X\to Y$:
\begin{enumerate}[label=(\arabic*)]
\item $f$ is open.
\item $\Path{f}$ is a p-morphism, \ie $\Path{f}(\up m)=\up \Path{f}(m)$ for all $m\in \Path{X}$.
\end{enumerate}
\end{prop}
\begin{proof}
$(1)\Rightarrow (2)$. Suppose $f$ is open, and let $m\colon P\emb X$ be an arbitrary element of $\Path{X}$. The inclusion $\Path{f}(\up m)\subseteq\up \Path{f}(m)$ follows at once from monotonicity of $\Path{f}$. For the converse inclusion, assume that $n\colon Q\emb Y$ is an element of $\Path{Y}$ above $\Path{f}(m)=f\circ m$. Then, the composite $f\circ m$ must factor through $n$, say $f\circ m=n\circ s$ for some embedding $s\colon P\emb Q$. Hence, we have a commutative square as displayed below.
\[\begin{tikzcd}
P \arrow[rightarrowtail]{r}{s} \arrow[rightarrowtail]{d}[swap]{m} & Q \arrow[rightarrowtail]{d}{n} \arrow[dashed,rightarrowtail]{dl}[swap,outer sep=-1pt, description]{m'} \\
X \arrow{r}{f} & Y
\end{tikzcd}\]
Since $f$ is open, there exists a diagonal filler $m'\colon Q\emb X$. The commutativity of the upper triangle entails that $m'\in\up m$, while the commutativity of the lower triangle implies that $\Path{f}(m')=n$. Therefore, $\up \Path{f}(m) \subseteq \Path{f}(\up m)$.

$(2)\Rightarrow (1)$. Assume $\Path{f}$ is a p-morphism and consider a commutative square as follows,
\[\begin{tikzcd}
P \arrow[rightarrowtail]{r}{s} \arrow[rightarrowtail]{d}[swap]{m} & Q \arrow[rightarrowtail]{d}{n} \\
X \arrow{r}{f} & Y
\end{tikzcd}\]
where $P$ and $Q$ are paths. We have $\Path{f}(m)\leq n$ in $\Path{Y}$, and thus there exists a path embedding $m'\colon P'\emb X$ satisfying $m\leq m'$ and $\Path{f}(m')=n$. The inequality $m\leq m'$ amounts to saying that $m=m'\circ l$ for some embedding $l\colon P\emb P'$, while the equality $\Path{f}(m')=n$ means that $f\circ m'=n\circ k$ for some isomorphism $k\colon P'\to Q$. We have a commutative diagram as displayed below.
\[\begin{tikzcd}
P \arrow[rightarrowtail]{dr}[swap]{m} \arrow[rightarrowtail]{r}{l} \arrow[rr, relay arrow=2ex, "s", rightarrowtail] & P' \arrow{r}{k} \arrow[rightarrowtail]{d}{m'} & Q \arrow[rightarrowtail]{d}{n} \\
{} & X \arrow{r}{f} & Y
\end{tikzcd}\]
The composite arrow $m'\circ k^{-1}\colon Q\to X$ satisfies $m'\circ k^{-1}\circ s=m$ and $f\circ m'\circ k^{-1}=n$, thus showing that $f$ is open. Just observe that the former equation holds by Lemma~\ref{l:path-cat-suprema}\ref{at-most-one-emb}, while the latter follows at once by diagram chasing.
\end{proof}

\begin{rem}\label{rem:arboreal-open-vs-JNW}
Concerning the notion of open morphism, and the ensuing notion of bisimulation (\cf Section~\ref{s:bisimulations}), our approach differs from the one adopted in \cite{JNW1993} in that our notion of path is completely determined by the factorisation system, and so is the notion of open morphism. While Lemma~\ref{l:open-pwemb-JNW-open} shows that the two notions of open morphism are equivalent for pathwise embeddings, in our setting it is useful to consider open morphisms and pathwise embeddings separately. For example, Proposition~\ref{p:pwe-open-dual-p-morph} would no longer hold if we dropped the requirement that $f$ be a pathwise embedding and replaced the notion of openness with the stronger one in \cite{JNW1993}. On the other hand, pathwise embeddings themselves have a game-theoretic and logical counterpart, \cf Proposition~\ref{pr:strong-exist-game} and Theorem~\ref{t:logic-fragm-character}.
\end{rem}

\subsection{Bisimulations}\label{s:bisimulations}
A \emph{bisimulation} between objects $X,Y$ of $\C$ is a span of open pathwise embeddings 
\[
X\leftarrow Z \rightarrow Y
\] 
in $\C$. If such a bisimulation exists, we say that $X$ and $Y$ are \emph{bisimilar}.

\begin{exas}
This definition directly generalizes that in \cite{AS2021}, and the notions of bisimulation given there for the Ehrenfeucht-\Fraisse, pebbling and modal comonads are the special cases arising in the categories $\RTk(\sg)$, $\RPk(\sg)$ and $\RMk(\sg)$ respectively, as described in Example~\ref{RTex}.
\end{exas}

\begin{rem}
Let $\C$ be a path category. If we regard trees as Kripke models where the accessibility relation is the tree order, then it follows from Theorem~\ref{thm:path-cat-functor-into-trees} and Proposition~\ref{p:pwe-open-dual-p-morph} that a span of pathwise embeddings 
\[\begin{tikzcd}[column sep=1.8em]
X & Z \arrow{l}[swap]{f} \arrow{r}{g} & Y
\end{tikzcd}\] 
in $\C$ is a bisimulation if, and only if, $\Path{X} \xleftarrow{\Path{f}} \Path{Z} \xrightarrow{\Path{g}} \Path{Y}$ is a bisimulation of Kripke models in the usual sense.
\end{rem}

Given a bisimulation $X\leftarrow Z \rightarrow Y$, we would like to think of $Z$ as providing a winning strategy for Duplicator in an appropriate game played ``between $X$ and $Y$''. To substantiate this idea, in Sections~\ref{s:arboreal} and~\ref{s:back-and-forth} we introduce \emph{arboreal categories}---a refinement of the concept of path category---and show that, in these categories, bisimilarity is captured by back-and-forth systems which model the dynamic nature of games.

\section{Arboreal Categories}\label{s:arboreal}
By Theorem~\ref{thm:path-cat-functor-into-trees}, any path category $\C$ admits a functor $\Path\colon \C\to\T$ into the category of (non-empty) trees. In general, the tree $\Path{X}$ may retain little information about $X$. We are interested in the case where $X$ is determined by $\Path{X}$. This leads us to the notion of \emph{path-generated} object.

Henceforth, we shall assume that any path category contains, up to isomorphism, only a set of paths. Thus, whenever convenient, we shall work with a small skeleton of the full subcategory of paths.

\subsection{Path-generated objects}
Let $\C$ be a path category. For any object $X$ of $\C$, we have a diagram with vertex $X$ consisting of all path embeddings with codomain $X$; the morphisms between paths are those which make the obvious triangles commute (note that they are necessarily embeddings):
\[\begin{tikzcd}[column sep=1.2em, row sep=2.2em]
 & X & \\
 P \arrow[bend left=20,rightarrowtail]{ur} \arrow[rightarrowtail]{rr} & & Q \arrow[bend right=20,rightarrowtail]{ul}
\end{tikzcd}\]
This can be seen as a cocone over the small diagram $\Path{X}$. We say that $X$ is \emph{path-generated} provided this is a colimit cocone in $\C$. Intuitively, an object is path-generated if it is the colimit of its paths.

Recall that a functor $J\colon \mathscr{A}\to \mathscr{B}$ is \emph{dense} if, for each $b\in\mathscr{B}$, the canonical cocone\footnote{This is the cocone whose component at
$(k\colon Ja \to b)\in J\down b$ is $k$ itself.} over the diagram 
\[\begin{tikzcd}[column sep=2em]
J\down b \arrow{r}{\pi} & \mathscr{A} \arrow{r}{J} & \mathscr{B},
\end{tikzcd}\] 
where $J\down b$ is the comma category and $\pi$ is the natural forgetful functor, is a colimit cocone. If $\Cp$ denotes the full subcategory of $\C$ defined by the paths, we have the following observation:
\begin{lem}\label{l:path-gen-dense}
The following statements are equivalent for any path category $\C$:
\begin{enumerate}[label=(\arabic*)]
\item Every object of $\C$ is path-generated.
\item The inclusion $\Cp\into \C$ is dense.
\end{enumerate}
\end{lem}
\begin{proof}
Fix an arbitrary object $X$ of $\C$ and consider the diagram 
\[\begin{tikzcd}[column sep=2em]
J\down X \arrow{r}{\pi} & \Cp \arrow{r}{J} & \C,
\end{tikzcd}\] 
where $J \colon \Cp\into \C$ is the inclusion functor. It suffices to show that $X$ is path-generated precisely when the canonical cocone $C$ with vertex $X$ on the diagram $J\circ \pi$, as depicted on the left-hand side below, is a colimit cocone.
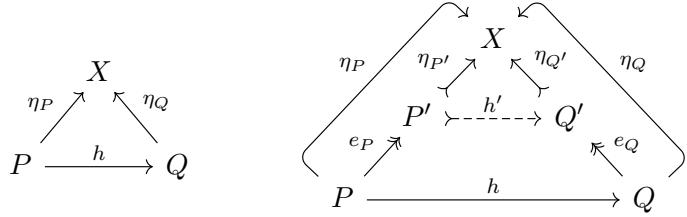
\begin{center}
\begin{tikzcd}[column sep=1em]
{} & X & \\
P \arrow{rr}{h} \arrow{ur}{\eta_P} & & Q \arrow[swap]{ul}{\eta_Q}
\end{tikzcd}
\ \ \ \ \ \ \ \ \ 
\begin{tikzcd}[column sep=0.8em, row sep=1.2em]
{} & & X & & \\
{} & P' \arrow[rightarrowtail,dashed]{rr}{h'} \arrow[rightarrowtail]{ur}{\eta_{P'}} & & Q' \arrow[rightarrowtail]{ul}[swap]{\eta_{Q'}} & \\
P \arrow[uurr, relay arrow=2.2ex, "\eta_P"] \arrow{rrrr}{h} \arrow[twoheadrightarrow]{ur}{e_P} & & & & Q \arrow[twoheadrightarrow]{ul}[swap]{e_Q} \arrow[uull, relay arrow=-2.2ex, "\eta_Q", swap]
\end{tikzcd}
\end{center}
That is, if $P=(J\circ \pi)(i)$ for some $i\in J\down X$ then $\eta_P=i$. Take the (quotient, embedding) factorisations of the morphisms in the cocone $C$, as displayed on the right-hand side above, and note that $P'$ and $Q'$ are paths by Lemma~\ref{path-images-and-subs}\ref{image-path}. Any arrow $h\colon P\to Q$ in the image of the diagram $J\circ\pi$ induces an embedding $h'\colon P'\emb Q'$ making the ensuing diagram commute. Just observe that the following square commutes,
\[\begin{tikzcd}
P \arrow[twoheadrightarrow]{r}{e_P} \arrow{d}[swap]{e_Q\circ h} & P' \arrow[rightarrowtail]{d}{\eta_{P'}} \\
Q' \arrow[rightarrowtail]{r}{\eta_{Q'}} & X
\end{tikzcd}\]
and so there is a diagonal filler $h'\colon P'\emb Q'$.

Now, suppose that $X$ is path-generated and consider a cocone $D$ with vertex $Y\in \C$ over the diagram $J\circ \pi$. This restricts to a cocone $D'$ on the diagram of paths that embed into $X$ and thus, since $X$ is path-generated, induces a unique arrow $f\colon X\to Y$ that is a morphism between the cocones $C'$ and $D'$, where $C'$ is the restriction of $C$ to the diagram of paths that embed into $X$.  It follows from the observations in the previous paragraph that $f$ is a morphism between the cocones $C$ and $D$. Just note that, if $\gamma_P$ is the component of $D$ at $P$,
\[
\gamma_P = \gamma_{P'} \circ e_P = f\circ \eta_{P'} \circ e_P = f\circ \eta_P
\] 
where in the first step we used the compatibility of $D$ combined with the fact that $e_P$ is the image of the functor $J\circ\pi$. Hence, $C$ is a colimit cocone.

Conversely, suppose that $C$ is a colimit cocone and consider a cocone $D'$ with vertex $Y\in \C$ over the diagram of paths that embed into $X$. Using the notation introduced in the first paragraph, $D'$ can be extended to a cocone $D$ over the diagram $J\circ \pi$ by precomposing with the appropriate quotients $e_P$. Compatibility of $D$ follows from compatibility of $D'$ and the fact that $e_Q\circ h = h'\circ e_P$. Thus, there is a unique arrow $f\colon X\to Y$ that is a morphism of cocones between $C$ and $D$. This clearly induces a morphism from the canonical cocone of path embeddings into $X$ to the cocone $D'$, and uniqueness of this morphism follows from uniqueness of $f$. Therefore, $X$ is path-generated.                                                                                                                                            
\end{proof}

\begin{rem}
Direct inspection shows that Lemma~\ref{l:path-gen-dense} holds for any category equipped with a stable proper factorisation system.
\end{rem}

We note in passing that, if every object of $\C$ is path-generated, the functor $\Path\colon \C\to\T$ is faithful on pathwise embeddings (recall that, by Lemma~\ref{l:path-cat-suprema}\ref{at-most-one-emb}, $\Path$ is always faithful on embeddings between paths).
To see this, suppose that $f,g\colon X\to Y$ are pathwise embeddings. If $\Path{f}=\Path{g}$ then, for all path embeddings $m\colon P\emb X$,
\[
f\circ m=\Path{f}(m)=\Path{g}(m)=g\circ m.
\]
As $X$ is path-generated, it follows that $f=g$.

\subsection{Arboreal categories} 
\begin{defi}\label{def:arboreal-cat}
An \emph{arboreal category} is a path category in which every object is path-generated.
\end{defi}

Equivalently, by Lemma~\ref{l:path-gen-dense}, an arboreal category is a path category $\C$ such that the inclusion $\Cp\into \C$ is dense. 

\begin{exas}
The category $\T$ of non-empty trees is arboreal; this is essentially the observation that every tree is the colimit of the diagram given by its branches and the embeddings between them. Similarly, $\F$ is arboreal.
Our key examples of the categories $\RTk(\sg)$, $\RPk(\sg)$ and $\RMk(\sg)$ from Example~\ref{RTex} are also arboreal.
\end{exas}

Note that, in view of Theorem~\ref{thm:path-cat-functor-into-trees}, any arboreal category $\C$ admits a functor $\Path\colon \C\to\T$ into the category of (non-empty) trees. This crucial fact is what will allow us, given an arboreal cover (\cf Section~\ref{s:arboreal-covers}), to	 process structures as tree-like objects. 

We collect some useful properties of arboreal categories.

\begin{lem}\label{l:arboreal-consequences}
Let $\C$ be an arboreal category. The following statements hold:
\begin{enumerate}[label=(\alph*)]
\item\label{canonical-suprema} For any object $X$ of $\C$ and any $m\in \Emb{X}$, $m=\bigvee{\{p\in\Path{X}\mid p\leq m\}}$. 
\item\label{quot-to-surj} If $f$ is a quotient in $\C$, then $\Path{f}$ is a surjection.
\end{enumerate} 
\end{lem}
\begin{proof}
For item~\ref{canonical-suprema}, let $m\colon S\emb X$ be an arbitrary embedding. Clearly, we have 
\[
\bigvee{\{p\in\Path{X}\mid p\leq m\}}\leq m.
\] 
For the converse direction, assume that $n\colon T\emb X$ is an upper bound for $\{p\in\Path{X}\mid p\leq m\}$. This means that each path embedding $P\emb X$ that factors through $m$ must factor through~$n$. To start with, consider the diagram of all path embeddings into $X$ that factor through $m$, as depicted on the left-hand side below.
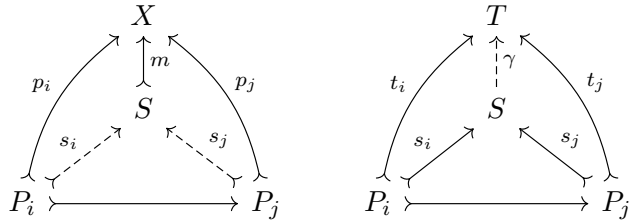
\begin{center}
\begin{tikzcd}
{} & X & \\
{} & S \arrow[rightarrowtail]{u}[swap]{m} & \\
P_i \arrow[rightarrowtail]{rr} \arrow[rightarrowtail, dashed]{ur}{s_i} \arrow[rightarrowtail, bend left=20]{uur}{p_i} & & P_j \arrow[rightarrowtail, dashed]{ul}[swap]{s_j}  \arrow[rightarrowtail, bend right=20]{uul}[swap]{p_j}
\end{tikzcd}
\ \ \ \ \ 
\begin{tikzcd}
{} & T & \\
{} & S \arrow[dashed]{u}[swap]{\gamma} & \\
P_i \arrow[rightarrowtail]{rr} \arrow[rightarrowtail]{ur}{s_i} \arrow[rightarrowtail, bend left=20]{uur}{t_i} & & P_j \arrow[rightarrowtail]{ul}[swap]{s_j} \arrow[rightarrowtail, bend right=20]{uul}[swap]{t_j}
\end{tikzcd}
\end{center}
Because $S$ is path-generated, the $s_i$'s form a colimit cocone and so $m$ is the unique arrow $S\to X$ making the left-hand diagram commute. Now, since each $p_i$ factors through $n$, say $p_i=n\circ t_i$ for some $t_i$, and $n$ is a monomorphism, the $t_i$'s form a compatible cocone. Using again the fact that the $s_i$'s form a colimit cocone, we get a unique mediating arrow $\gamma\colon S\to T$ making the right-hand diagram above commute. By uniqueness of $m$ we have $n\circ\gamma = m$, and so $m\leq n$.

For item~\ref{quot-to-surj}, suppose that $f\colon X\epi Y$ is a quotient in $\C$. We first assume that $Y=P$ is a path, and then settle the general case. To show that $\Path{f}$ is a surjection, it suffices to prove that $\id_P\in \Path{f}(\Path{X})$, where $\id_P\colon P\to P$ is the identity. We have
\begin{align*}
f^* \id_P &= \bigvee{\{p\in\Path{X}\mid p\leq f^* \id_P\}} \tag*{item~\ref{canonical-suprema}} \\
&= \bigvee{\{p\in\Path{X}\mid \exists_f p\leq \id_P\}} \tag*{Lemma~\ref{l:adj-image-pullback}} \\
&= \bigvee{\Path{X}},
\end{align*}
and so (using the fact that left adjoints preserve suprema)
\begin{align*}
\id_P = \exists_f f^* \id_P &= \bigvee{\Path{f} (\Path{X})},
\end{align*}
where the first step follows from the fact that $\exists_f \circ f^*$ is the identity of $\Emb{P}$ (\cf the proof of Lemma~\ref{l:emb-quo-order-embeddings}\ref{pullback-embed}).
Hence, $\id_P\in \Path{f}(\Path{X})$ by Lemma~\ref{l:path-cat-suprema}\ref{paths-j-irred}.

For the general case, let $m\colon P\emb Y$ be an arbitrary path embedding and consider the following pullback square in $\C$. 
\[\begin{tikzcd}
f^* P \arrow[twoheadrightarrow]{r}{g} \arrow[rightarrowtail]{d}[swap]{f^* m} \arrow[dr, phantom, "\lrcorner", very near start] & P \arrow[rightarrowtail]{d}{m} \\
X \arrow[twoheadrightarrow]{r}{f} & Y
\end{tikzcd}\]
By the argument above, there is a path embedding $n\colon Q\emb f^*P$ such that $\Path{g}(n)=\id_P$. It follows that $\Path{f}(f^*m\circ n)=m$. Just observe that the following square commutes, and $g\circ n$ is a quotient because $\Path{g}(n)=\id_P$.
\[\begin{tikzcd}
Q \arrow[twoheadrightarrow]{r}{g\circ n} \arrow[rightarrowtail]{d}[swap]{f^* m\, \circ \, n} & P \arrow[rightarrowtail]{d}{m} \\
X \arrow[twoheadrightarrow]{r}{f} & Y
\end{tikzcd}\]
Since factorisations are unique up to isomorphism, we conclude that $\Path{f}(f^*m\circ n)=m$.
\end{proof}

Let $X$ be an object of an arboreal category. Given a subset $\U$ of $\Path{X}$, consider its supremum $n\colon \bigvee{\U}\emb X$ in $\Emb{X}$. In view of Lemma~\ref{l:emb-quo-order-embeddings}\ref{exists-embed}, $\exists_n\colon \Emb{(\bigvee{\U})} \emb \Emb{X}$ is an order-embedding, and so is its restriction $\Path{n}\colon \Path{(\bigvee\U)}\to \Path{X}$. Thus, we shall identify $\Path{(\bigvee\U)}$ with its image in $\Path{X}$. Note that the downward closure of $\U$ is contained in $\Path{(\bigvee\U)}$.
The next proposition shows that, in fact, $\Path{(\bigvee\U)}$ is \emph{equal to} the downward closure of $\U$. This will allow us to construct an object from a prescribed set of path embeddings, without adding any new paths in the process. 
\begin{prop}\label{pr:perfect-lattice-of-strong-subs}
Let $\C$ be an arboreal category, $X$ an object of $\C$, and $\U\subseteq \Path{X}$ a non-empty subset. A path embedding $m\in\Path{X}$ is below $\bigvee{\U}$ if, and only if, it is below some element of~$\U$.
\end{prop}
\begin{proof}
Fix an arbitrary object $X$ of $\C$ and a non-empty set of path embeddings $\U=\{m_i\colon P_i\emb X\mid i\in I\}$. Let $m\colon P\emb X$ be an arbitrary path embedding. If $m$ is below some element of $\U$, then clearly $m\leq \bigvee\U$. 

For the converse direction, suppose that $m\leq \bigvee\U$. Recall from the proof of Lemma~\ref{l:path-cat-suprema}\ref{suprema-of-paths} that the supremum of $\U$ is obtained by taking the (quotient, embedding) factorisation $\coprod_{i\in I}P_i\xepi{e} S\xemb{n} X$ of the canonical morphism $\coprod_{i\in I}P_i\to X$. With this notation, $\bigvee \U=n$. Since $m\leq \bigvee\U$, there exists an embedding $m'\colon P\emb S$ such that $m=n\circ m'$. Consider the pullback of $m'$ along $e$:
\[\begin{tikzcd}
T \arrow[rightarrowtail]{d}[swap]{j} \arrow[twoheadrightarrow]{r}{r} \arrow[dr, phantom, "\lrcorner", very near start] & P \arrow[rightarrowtail]{d}{m'} \\
 \coprod_{i\in I}P_i \arrow[twoheadrightarrow]{r}{e} & S
\end{tikzcd}\]
Applying Lemma~\ref{l:arboreal-consequences}\ref{quot-to-surj} to the quotient $r$, we see that there exists a path embedding $k\colon Q\emb T$ such that $\Path{r}(k)=\id_P$, \ie $r\circ k$ is a quotient. Because $Q$ is connected, $j\circ k\colon Q\emb  \coprod_{i\in I}P_i$ must factor through some coproduct injection $\phi_i\colon P_i\to \coprod_{i\in I}P_i$, \iec $j\circ k= \phi_i\circ p$ for some path embedding $p\colon Q\emb P_i$. We then have a commutative diagram as follows.
\[\begin{tikzcd}[column sep=3em, row sep=2.5em]
Q \arrow[twoheadrightarrow]{r}{r\circ k} \arrow[rightarrowtail]{d}[swap]{p} & P \arrow[rightarrowtail]{d}[swap]{m'} \arrow[rightarrowtail]{dr}{m} & {} \\
P_i \arrow{r}{e\circ \phi_i} \arrow[rr, relay arrow=-2ex, "m_i", swap, rightarrowtail] & S \arrow[rightarrowtail]{r}{n} & X
\end{tikzcd}\]
As $m\circ r\circ k=m_i\circ p$ and the right-hand side of the equation is an embedding, $r\circ k$ is an isomorphism. So $m\leq m_i\in \U$, thus concluding the proof.
\end{proof}

\begin{rem}
Let $X$ be any object of an arboreal category $\C$. Lemma~\ref{l:path-cat-suprema}\ref{suprema-of-paths}, combined with Lemma~\ref{l:arboreal-consequences}\ref{canonical-suprema}, implies that the poset of embeddings $\Emb{X}$ admits all suprema, hence is a complete lattice. In fact, $\Emb{X}$ is a \emph{perfect} lattice. 

Recall that an element $x$ of a complete lattice $L$ is \emph{completely join-irreducible} provided that, for any subset $M \subseteq L$, $x=\bigvee M$ entails $x\in M$. A lattice is perfect if it is complete, completely distributive, and each of its elements is a supremum of completely join-irreducible ones. It is well known, \cf \cite{Raney1952} or~\cite[Theorem~10.29]{DP2002}, that a lattice is perfect precisely when it is isomorphic to the lattice $\mathcal{D}(P)$ of downwards closed subsets of a poset $P$. Equivalently, when it is isomorphic to the lattice of downwards closed subsets of the poset of its completely join-irreducible elements.

For an object $X$ of an arboreal category, let $F$ be the forest obtained by removing the root from the tree $\Path{X}$ (note that the least element of a complete lattice is never completely join-irreducible, since it coincides with the supremum of the empty set). Lemma~\ref{l:arboreal-consequences}\ref{canonical-suprema} and Proposition~\ref{pr:perfect-lattice-of-strong-subs} entail that the map
\[
\Emb{X} \to \mathcal{D}(F), \ \ m\mapsto \{p\in F \mid p\leq m\}
\]
is a bijection, whose inverse sends a downwards closed subset of $F$ to its supremum in $\Emb{X}$.  Both assignments are clearly monotone, and since any order isomorphism between lattices is a lattice isomorphism, the map above is an isomorphism of (complete) lattices. Thus, in an arboreal category, the lattice of embeddings of any object is isomorphic to the lattice of downwards closed subsets of some forest.
\end{rem}

\section{Back-and-Forth Systems and Games}\label{s:back-and-forth}
Throughout this section, we work in a fixed arboreal category $\C$. First, we introduce back-and-forth systems in $\C$ and show that they capture precisely the bisimilarity relation defined in Section~\ref{s:openmaps} in terms of spans of open pathwise embeddings. Then, we show that back-and-forth systems can be equivalently seen as appropriate back-and-forth games.
\subsection{Back-and-forth systems}
Given objects $X$ and $Y$ of $\C$, we consider spans of (equivalence classes of) path embeddings of the form
$X\xlemb{m} P \xemb{n} Y$. Such a span can be thought of as a partial isomorphism ``of shape $P$'' between $X$ and $Y$. A back-and-forth system between $X$ and $Y$ is a collection of such spans containing an ``initial element'' and satisfying an appropriate extension property. 

Let $X,Y$ be any two objects of $\C$. Given $m\in\Path{X}$ and $n\in\Path{Y}$, we write $\br{m,n}$ to denote that $\dom(m)\cong \dom(n)$. Observe that (i) any two embeddings in the same $\sim$-equivalence class have isomorphic domains (where $\sim$ is the equivalence relation defined on p.~\pageref{symmetriz-preorder-subobjects}), and (ii) given $\br{m,n}$, there exist path embeddings $m'\sim m$ and $n'\sim n$ such that $\dom(m')=\dom(n')$. Hence, the pairs of the form $\br{m,n}$ capture the partial isomorphisms $X\xlemb{m} P \xemb{n} Y$ ``of shape $P$''. 

\begin{defi}\label{def:back-and-forth}
A \emph{back-and-forth system} between objects $X,Y$ of $\C$ is a set 
\[
\B=\{\br{m_i,n_i}\mid m_i\in\Path{X}, \, n_i\in\Path{Y}, \, i\in I\}
\] 
satisfying the following conditions:
\begin{enumerate}[label=(\roman*)]
\item\label{initial} $\br{\bot_X,\bot_Y}\in\B$, where $\bot_X$ and $\bot_Y$ are the roots of $\Path{X}$ and $\Path{Y}$, respectively.
\item\label{forth} If $\br{m,n}\in\B$ and $m'\in\Path{X}$ are such that $m\cvr m'$, there exists $n'\in\Path{Y}$ satisfying $n\cvr n'$ and $\br{m',n'}\in\B$.
\item\label{back} If $\br{m,n}\in\B$ and $n'\in\Path{Y}$ are such that $n\cvr n'$, there exists $m'\in\Path{X}$ satisfying $m\cvr m'$ and $\br{m',n'}\in\B$.
\end{enumerate}
A back-and-forth system $\B$ is \emph{strong} if, for all $\br{m,n}\in \B$ and all $m'\in\Path{X}$, $n'\in\Path{Y}$, if $m'\cvr m$ and $n'\cvr n$ then $\br{m',n'}\in\B$.
\end{defi}

Two objects $X$ and $Y$ of $\C$ are said to be \emph{(strong) back-and-forth equivalent} if there exists a (strong) back-and-forth system between them.

\begin{rem}
The definition of (strong) back-and-forth system given above is a variant of the notion of \emph{(strong) path bisimulation} from~\cite{JNW1996}. The nomenclature adopted here is motivated by the analogy with back-and-forth systems of partial isomorphisms from model theory \cite{Libkin2004}.
\end{rem}

\begin{lem}\label{l:always-strong}
Any two objects of an arboreal category are back-and-forth equivalent if, and only if, they are strong back-and-forth equivalent.
\end{lem}
\begin{proof}
Let $X,Y$ be objects of an arboreal category. For the non-trivial direction, suppose that there exists a back-and-forth system 
\[
\B=\{\br{m_i,n_i}\mid m_i\in\Path{X}, \, n_i\in\Path{Y}, \, i\in I\}
\] 
between $X$ and $Y$. To start with, note that for all $\br{m,n}\in\B$ the height of $m$ in $\Path{X}$ coincides with the height of $n$ in $\Path{Y}$. To see this, assume without loss of generality that $\dom(m) = P = \dom(n)$. Then the chain ${\downarrow} m$ of elements of $\Path{X}$ that are below $m$ is isomorphic to $\Emb{P}$, and similarly for ${\downarrow} n$.

Given $\br{m,n}\in\B$, suppose that $m$ and $n$ have height $k$. For every $j\in\{1,\ldots,k\}$, denote by $m_j$ and $n_j$ the unique path embeddings of height $j$ below $m$ and $n$, respectively. If $\br{m_j,n_j}\in\B$ for all $j\in\{1,\ldots,k\}$, the pair $\br{m,n}$ is said to be \emph{reachable}. We claim that
\[
\B'\coloneqq \{\br{m,n}\in\B \mid \br{m,n} \ \text{is reachable}\}
\]
is a strong back-and-forth system between $X$ and $Y$. 

Item~\ref{initial} in Definition~\ref{def:back-and-forth} is satisfied because $\br{\bot_X,\bot_Y}$ belongs to $\B$ and is trivially reachable. For item~\ref{forth}, suppose that $\br{m,n}\in\B'$, and let $m'\in\Path{X}$ satisfy $m\cvr m'$. Since $\B$ is a back-and-forth system, there is $n'\in\Path{Y}$ such that $n\cvr n'$ and $\br{m',n'}\in\B$. As $\br{m,n}$ is reachable, so is $\br{m',n'}$. Item~\ref{back} follows by a similar reasoning.
Finally, note that $\B'$ is a \emph{strong} back-and-forth system by construction.
\end{proof}

The aim of this section is to clarify the relation between bisimilarity and back-and-forth equivalence in arboreal categories. While it turns out that any two bisimilar objects $X$ and $Y$ are back-and-forth equivalent, to construct a bisimulation from a back-and-forth system we will assume that the product $X\times Y$ exists. For convenience of notation, given any two objects $a,b$ of a category $\mathscr{A}$, let us write $a \pro b$ to denote that the product $a\times b$ exists in $\mathscr{A}$. The following is our main result on bisimilarity:
\begin{thm}\label{th:bisimilar-iff-strong-back-forth}
Let $X,Y$ be any two objects of an arboreal category such that $X\pro Y$.
Then $X$ and $Y$ are bisimilar if, and only if, they are back-and-forth equivalent.
\end{thm}

In particular, in any arboreal category with binary products, the bisimilarity relation coincides with back-and-forth equivalence. This is the case in all our key examples:
\begin{exas}
The category $\T$ of non-empty trees has binary products. These are computed as ``synchronous products'' consisting of the pairs $(x,y)$ of elements having the same height, with the componentwise order. Similarly, $\F$ has binary products, and so does the category $\R(\sg)$ of forest-ordered $\sg$-structures if the relations in the synchronous product are interpreted componentwise. As synchronous products do not increase the height of forests, we see that the category $\RTk(\sg)$ from Example~\ref{RTex} has binary products. The category $\RPk(\sg)$ has also finite products, but for these we need to restrict to those pairs $(x,y)$ of elements with same height on which the pebbling functions coincide. Finally, binary products exist in $\RMk(\sg)$ and can be described again as variants of synchronous products.\footnote{Alternatively, as pointed out in~\cite{AS2021}, $\RMk(\sg)$ can be identified with a coreflective subcategory of the category of \emph{pointed} Kripke structures. Since the latter is complete, so is $\RMk(\sg)$.}
\end{exas}

We start by establishing the easy direction of Theorem~\ref{th:bisimilar-iff-strong-back-forth}, whose proof is akin to that of \cite[Lemma~16]{JNW1996}.
\begin{prop}\label{p:bisimilar-implies-back-and-forth}
Any two bisimilar objects of an arboreal category are back-and-forth equivalent.
\end{prop}
\begin{proof}
Suppose that $X\xleftarrow{f} Z \xrightarrow{g} Y$ is a span of open pathwise embeddings in an arboreal category $\C$. We claim that 
\[
\B\coloneqq \{\br{\Path{f}(m),\Path{g}(m)}\mid m\in\Path{Z}\}
\]
is a back-and-forth system between $X$ and $Y$ (in fact, this is a strong back-and-forth system). We show that items~\ref{initial}, \ref{forth} and~\ref{back} in Definition~\ref{def:back-and-forth} are satisfied.

For item~\ref{initial}, let $\bot_Z$ be the root of $\Path{Z}$. Then $\Path{f}(\bot_Z)=\bot_X$ and $\Path{g}(\bot_Z)=\bot_Y$ because tree morphisms preserve roots, and so $\br{\bot_X,\bot_Y}\in \B$.

For item~\ref{forth}, let $\br{\Path{f}(m),\Path{g}(m)}\in\B$ and $m'\in\Path{X}$ be such that $\Path{f}(m)\cvr m'$. Let us denote $P\coloneqq \dom(m)$ and $P'\coloneqq \dom(m')$. As $f\circ m\leq m'$ in $\Path{X}$, there exists $k\colon P\emb P'$ such that $f\circ m=m'\circ k$ in $\C$. Therefore, we have a commutative square as follows.
\[\begin{tikzcd}
P \arrow[rightarrowtail]{r}{k} \arrow[rightarrowtail]{d}[swap]{m} & P' \arrow[rightarrowtail]{d}{m'}  \\
Z \arrow{r}{f} & X
\end{tikzcd}\]
Since $f$ is open, there is a diagonal filler $n\colon P'\emb Z$. Thus, $m'=\Path{f}(n)$ and $\br{m',\Path{g}(n)}\in\B$.
It remains to show that $\Path{g}(m)\cvr \Path{g}(n)$. Note that $m\leq n$ and $\Path{f}(m)\cvr \Path{f}(n)$ entail $m\cvr n$, and so $\Path{g}(m)\cvr \Path{g}(n)$ because $\Path{g}$ preserves the covering relation.

The proof of item~\ref{back} is the same, mutatis mutandis, as for~\ref{forth}.
\end{proof}

\begin{rem}
Suppose for a moment that we define a back-and-forth system between $X$ and $Y$ to consist of spans of path embeddings $X \lemb P \emb Y$, rather than \emph{equivalence classes} of such path embeddings. Then Proposition~\ref{p:bisimilar-implies-back-and-forth} fails to hold. The reason is that, for a bisimulation 
\[\begin{tikzcd}[column sep=1.8em]
X & Z \arrow{l}[swap]{f} \arrow{r}{g} & Y,
\end{tikzcd}\] 
the obvious candidate 
\[
\{\begin{tikzcd}[column sep=2.2em] \hspace{-0.4em}X  & P \arrow[rightarrowtail]{l}[swap]{f\circ m} \arrow[rightarrowtail]{r}{g\circ m} & Y \hspace{-0.4em}\end{tikzcd} \mid m\colon P\emb Z \ \text{is a path embedding}\}
\]
need not be a set, in general. Even the notion of \emph{strong} back-and-forth system would be problematic, as there may be a proper class of path embeddings with codomain $P$. Of course, in the usual applications the subcategory of paths has a small skeleton, which allows to circumvent this problem.
\end{rem}

To establish the other direction of Theorem~\ref{th:bisimilar-iff-strong-back-forth}, we start by considering a (strong) back-and-forth system $\B=\{\br{m_i,n_i}\mid i\in I\}$ between $X$ and $Y$, and attempt to construct an object $Z$ and a span of open pathwise embeddings $X\leftarrow Z\rightarrow Y$. Intuitively, $Z$ is obtained by gluing together the paths $P_i\coloneqq \dom(m_i)$, for $i\in I$, by taking a colimit in $\C$. 
If the product of $X$ and $Y$ exists in $\C$, then this colimit can be equivalently described as the supremum of a set of path embeddings as we now explain. 

Consider an arbitrary $\br{m_i,n_i}\in\B$ and assume without loss of generality that $\dom(m_i)=P_i=\dom(n_i)$ for some path $P_i$. 
Then the product arrow 
\[
\langle m_i,n_i\rangle \colon P_i\to X\times Y
\] 
is an embedding. In fact, it suffices that $m_i$ be an embedding (or, symmetrically, that $n_i$ be an embedding), for then $m_i=\pi_X\circ \langle m_i,n_i\rangle$ entails that $\langle m_i,n_i \rangle$ is an embedding, where $\pi_X\colon X\times Y\to X$ is the projection.
Therefore, we can identify each $\br{m_i,n_i}\in\B$ with a path embedding $\langle m_i,n_i\rangle\in \Path{(X\times Y)}$ and compute the supremum $m\colon Z\emb X\times Y$ in $\Emb{(X\times Y)}$ of all these path embeddings. Note that the assignment 
\[
\br{m_i,n_i}\mapsto \langle m_i,n_i \rangle\in \Path{(X\times Y)}
\] 
does not depend on the choice of the representatives in the equivalence classes of $m_i$ and $n_i$. We note in passing the following immediate fact:
\begin{lem}\label{l:strong-baf-system-downwards}
Let $\B=\{\br{m_i,n_i}\mid i\in I\}$ be a back-and-forth system between $X$ and $Y$. If $X\pro Y$ and $\B$ is strong, $\{\langle m_i,n_i\rangle \in \Path{(X\times Y)} \mid i\in I\}$ is downwards closed in $\Emb{(X\times Y)}$.\qed
\end{lem}
To show that the span $X\xleftarrow{\pi_X\circ m} Z\xrightarrow{\pi_Y\circ m} Y$ is a bisimulation, we exploit the fact that $Z$ does not admit more path embeddings than those prescribed (\cf Proposition~\ref{pr:perfect-lattice-of-strong-subs}). The following proposition then completes the proof of Theorem~\ref{th:bisimilar-iff-strong-back-forth}:

\begin{prop}\label{p:b-and-f-implies-bisimilar}
Let $X$ and $Y$ be any two objects of an arboreal category such that $X\pro Y$. If $X$ and $Y$ are back-and-forth equivalent, then they are bisimilar.
\end{prop}
\begin{proof}
Let $\C$ be an arboreal category, and let $X,Y$ be any two back-and-forth equivalent objects of $\C$ admitting a product $X\times Y$. By Lemma~\ref{l:always-strong}, there is a strong back-and-forth system $\B=\{\br{m_i,n_i}\mid i\in I\}$ between $X$ and $Y$. Consider the set
\[
\U\coloneqq\{\langle m_i,n_i\rangle \in \Path{(X\times Y)} \mid i\in I\}
\]
and let $m\colon Z\emb X\times Y$ be the supremum of $\U$ in $\Emb{(X\times Y)}$. We claim that 
\[
\begin{tikzcd}[column sep=3em]
X & Z \arrow{r}{\pi_Y\circ m} \arrow{l}[swap]{\pi_X\circ m} & Y
\end{tikzcd}
\]
is a bisimulation between $X$ and $Y$.

To see that this is a span of pathwise embeddings, consider an arbitrary path embedding $n\colon P\emb Z$. In view of Proposition~\ref{pr:perfect-lattice-of-strong-subs} and Lemma~\ref{l:strong-baf-system-downwards}, $m\circ n\in \U$. That is, $m\circ n=\langle m_i,n_i\rangle$ in $\Path{(X\times Y)}$ for some $\br{m_i,n_i}\in\B$. It follows that $m\circ n=\langle m_i,n_i\rangle\circ \phi$ in $\C$ for some isomorphism $\phi$, and so $\pi_X\circ m\circ n$ and $\pi_Y\circ m\circ n$ are embeddings because $\pi_X\circ m\circ n=m_i\circ \phi$ and $\pi_Y\circ m\circ n=n_i\circ\phi$.

It remains to show that $\pi_X\circ m$ and $\pi_Y\circ m$ are open. We prove that $\pi_X\circ m$ is open; the proof for $\pi_Y\circ m$ follows by symmetry. Consider a commutative square in $\C$ as displayed below, where $P$ and $Q$ are paths.
\[\begin{tikzcd}[column sep=3.0em, row sep=2.5em]
P \arrow[rightarrowtail]{r}{k} \arrow[rightarrowtail]{d}[swap]{n} & Q \arrow[rightarrowtail]{d}{m_j} \\
Z \arrow{r}{\pi_X\circ m} & X
\end{tikzcd}\]
Reasoning as above, we see that $m\circ n=\langle m_i, n_i\rangle$ in $\Path{(X\times Y)}$ for some $\br{m_i,n_i}\in\B$. Therefore, in $\Path{X}$, we have $m_i=\pi_X\circ m\circ n\leq m_j$. Applying item~\ref{forth} in Definition~\ref{def:back-and-forth} (possibly finitely many times), it follows that there exists $n_j\in\Path{Y}$ such that $n_i\leq n_j$ and $\br{m_j,n_j}\in \B$. Suppose without loss of generality that $\dom(m_j)=\dom(n_j)$. Then $\langle m_j,n_j\rangle \in\U$ and so $\langle m_j,n_j\rangle\colon Q\emb X\times Y$ factors through the supremum $m\colon Z\emb X\times Y$ of $\U$. That is, $\langle m_j,n_j\rangle=m\circ h$ in $\C$ for some morphism $h\colon Q\to Z$. We claim that the following diagram commutes, thus establishing that $\pi_X\circ m$ is open.
\[\begin{tikzcd}[column sep=3.0em, row sep=2.5em]
P \arrow[rightarrowtail]{r}{k} \arrow[rightarrowtail]{d}[swap]{n} & Q \arrow[rightarrowtail]{d}{m_j} \arrow[outer sep=-1pt]{dl}[description]{h} \\
Z \arrow{r}{\pi_X\circ m} & X
\end{tikzcd}\]
For the commutativity of the lower triangle, just observe that
\[
\pi_X\circ m\circ h= \pi_X\circ \langle m_j,n_j\rangle=m_j.
\]
Now, assume without loss of generality that $\dom(m_i)=R=\dom(n_i)$ for some path $R$. As already observed above, $m\circ n=\langle m_i, n_i\rangle$ in $\Path{(X\times Y)}$ implies that $m\circ n=\langle m_i,n_i\rangle \circ \phi$ in $\C$ for some isomorphism $\phi\colon P\to R$. Thus, for the upper triangle, we have 
\begin{align*}
n=h\circ k \ &\Longleftrightarrow \ m\circ n= m\circ h\circ k \\
&\Longleftrightarrow \ \langle m_i, n_i\rangle \circ \phi=\langle m_j, n_j\rangle \circ k \\
&\Longleftrightarrow \ \begin{cases} m_i\circ \phi=m_j\circ k \\ n_i\circ \phi=n_j\circ k\end{cases}
\end{align*}
where in the first step we used the fact that $m$ is a monomorphism.
In turn, the inequalities $m_i\leq m_j$ and $n_i\leq n_j$ entail the existence of embeddings $k_1,k_2\colon R\emb Q$ such that $m_i=m_j\circ k_1$ and $n_i=n_j\circ k_2$. By Lemma~\ref{l:path-cat-suprema}\ref{at-most-one-emb} we have $k_1\circ\phi=k=k_2\circ\phi$. It follows that $m_i\circ\phi=m_j\circ k$ and $n_i\circ\phi=n_j\circ k$, and so $n=h\circ k$.
\end{proof}

\subsection{Back-and-forth games}
Let $\C$ be an arboreal category and let $X,Y$ be any two objects of $\C$. We define a back-and-forth game $\G(X,Y)$ played by Spoiler and Duplicator on $X$ and $Y$ as follows.
Positions in the game are pairs of (equivalence classes of) path embeddings $(m,n)\in\Path{X}\times\Path{Y}$.
The winning relation $\W(X,Y)\subseteq \Path{X}\times\Path{Y}$ consists of the pairs $(m,n)$ such that $\dom(m)\cong\dom(n)$.

Let $\bot_X\colon P\emb X$ and $\bot_Y\colon Q\emb Y$ be the roots of $\Path{X}$ and $\Path{Y}$, respectively. If $P\not\cong Q$, then Duplicator loses the game. Otherwise, the initial position is $(\bot_X,\bot_Y)$.
At the start of each round, the position is specified by a pair $(m,n)\in\Path{X}\times\Path{Y}$, and the round proceeds as follows: Either Spoiler chooses some $m'\succ m$ and Duplicator must respond with some $n'\succ n$, or Spoiler chooses some $n''\succ n$ and Duplicator must respond with $m''\succ m$. Duplicator wins the round if they are able to respond and the new position is in $\W(X,Y)$. Duplicator wins the game if they have a strategy $\Sigma$ such that, for all $t\geq 0$, $\Sigma$ is winning after $t$ rounds.

\begin{exas}
It is shown in \cite{AS2021} that the abstract game $\G(X,Y)$ coincides, in the case of the arboreal categories $\RTk(\sg)$, $\RPk(\sg)$ and $\RMk(\sg)$, with the usual $k$-round Ehrenfeucht-\Fraisse, $k$-pebble, and $k$-round bisimulation games, respectively.
\end{exas}

The following straightforward observation makes precise the translation between back-and-forth systems and back-and-forth games.
\begin{lem}\label{l:back-and-forth-sys-games}
Two objects $X,Y$ of an arboreal category $\C$ are back-and-forth equivalent if, and only if, Duplicator has a winning strategy in the game $\G(X,Y)$.
\end{lem}
\begin{proof}
Clearly, if $\B=\{\br{m_i,n_i}\mid i\in I\}$ is a back-and-forth system between $X$ and $Y$, then the plays in the set
\[
\{(m_i, n_i) \mid i\in I\}\subseteq \Path{X}\times \Path{Y}
\]
yield a winning strategy for Duplicator in the game $\G(X,Y)$.

Conversely, a winning strategy for Duplicator in the game $\G(X,Y)$ determines a set $W\subseteq \W(X,Y)$ of the plays following this strategy, and 
\[
\B\coloneqq \{\br{m,n}\mid (m,n)\in W\}
\]
is a back-and-forth system. In fact, $\B$ is even strong because $W$ is downwards closed in the (synchronous) product of trees $\Path{X}\times \Path{Y}$.
\end{proof}

The previous lemma, combined with Theorem~\ref{th:bisimilar-iff-strong-back-forth}, yields at once the following result:
\begin{thm}\label{th:gamesiffbisim}
Let $X,Y$ be objects of an arboreal category such that $X\pro Y$. Then $X$ and~$Y$ are bisimilar if, and only if, Duplicator has a winning strategy in the game $\G(X,Y)$.  \qed
\end{thm}

We record an easy but useful consequence\footnote{Note that, alternatively, we could have deduced Corollary~\ref{cor:bisim-equiv-rel} from Theorem~\ref{th:bisimilar-iff-strong-back-forth}.} of the previous theorem:
\begin{cor}\label{cor:bisim-equiv-rel}
In an arboreal category, suppose that $X\pro Y$, $Y\pro Z$ and $X\pro Z$. If $X$ and~$Y$, and $Y$ and $Z$, are bisimilar, then so are $X$ and $Z$.
\end{cor}
\begin{proof}
This is a consequence of Theorem~\ref{th:gamesiffbisim}. Just observe that, if Duplicator has a winning strategy for the games $\G(X,Y)$ and $\G(Y,Z)$, then the composition of these strategies is a winning strategy for the game $\G(X,Z)$.

In more detail, let $\bot_X,\bot_Y$ and $\bot_Z$ be the roots of $\Path{X},\Path{Y}$ and $\Path{Z}$, respectively. Then $(\bot_X,\bot_Y)\in \W(X,Y)$ and $(\bot_Y,\bot_Z)\in \W(Y,Z)$ entail $\dom(\bot_X)\cong \dom(\bot_Y)\cong \dom(\bot_Z)$, and so $(\bot_X,\bot_Z)\in \W(X,Z)$. After $i$ rounds of the game, the position is given by a pair $(m,p)\in\Path{X}\times \Path{Z}$. By inductive hypothesis, we can assume that there exists $n\in \Path{Y}$ such that $(m,n)\in \W(X,Y)$ and $(n,p)\in\W(Y,Z)$. Suppose that, in the following round, Spoiler chooses some $m'\succ m$. Following their winning strategies for the games $\G(X,Y)$ and $\G(Y,Z)$, Duplicator can find $n'\succ n$ and $p'\succ p$ such that $(m',n')\in\W(X,Y)$ and $(n',p')\in\W(Y,Z)$, and so $(m',p')\in\W(X,Z)$. Therefore, Duplicator will respond with $p'$. If Spoiler chooses some $p''\succ p$, the argument is the same, mutatis mutandis.
\end{proof}

\begin{rem}
In view of Corollary~\ref{cor:bisim-equiv-rel}, bisimilarity is an equivalence relation in any arboreal category admitting binary products.
\end{rem}

There is also an \emph{existential positive} version of the game $\G(X,Y)$, denoted by $\EG(X,Y)$, where Spoiler always plays in $\Path{X}$, and Duplicator always responds in $\Path{Y}$. The winning relation $\W(X,Y)\subseteq \Path{X}\times\Path{Y}$ consists of the pairs $(m,n)$ such that there exists a morphism $\dom(m)\to\dom(n)$. 

\begin{prop}\label{pr:exist-game}
Let $\C$ be an arboreal category and let $X,Y$ be any two objects of $\C$. If $\Path\colon\C\to\T$ is faithful, then Duplicator has a winning strategy in the game $\EG(X,Y)$ if, and only if, there exists a morphism $X\to Y$.  
\end{prop}
\begin{proof}
Suppose that there exists a morphism $f\colon X\to Y$ in $\C$. As $\Path{f}$ is a tree morphism,
\begin{equation}\label{eq:mor-winning-st}
\{(m, \Path{f}(m)) \mid m\in\Path{X}\}\subseteq \Path{X}\times \Path{Y}
\end{equation}
is a winning strategy for Duplicator in the game $\EG(X,Y)$. 

Conversely, we show that a winning strategy for Duplicator in the game $\EG(X,Y)$ determines a cocone with vertex $Y$ over the diagram of paths that embed into $X$. Since $X$ is path-generated, this cocone induces a morphism $X\to Y$. We explain in detail how the cocone with vertex $Y$ is constructed.

Fixing representatives, we can regard $\Path{X}$ as a cocone in $\C$ over a small diagram $D$ consisting of paths and path embeddings. For a generic path embedding $m\colon P\emb X$ in $D$, the cocone 
\[
E=\{e_m\colon P\to Y \mid m\in D\} 
\]
is defined by induction on the height of $m$ as follows:
\begin{enumerate}[label=(\roman*)]
\item\label{base-case} For the base case, suppose that $m=\bot_X\colon P\emb X$ is the root of $\Path{X}$, and let $\bot_Y\colon Q\emb Y$ be the root of $\Path{Y}$. Because $(\bot_X,\bot_Y)\in\W(X,Y)$, there is a morphism $h\colon P\to Q$. Set $e_m\coloneqq \bot_Y\circ h$.
\item\label{inductive-step} Assume that, for all $m\in D$ of height at most $k$ in $\Path{X}$, $e_m$ has been defined such that $e_m=n\circ h$ for some $n\in\Path{Y}$ satisfying $(m, n)\in\W(X,Y)$. If $m'\in D$ has height $k+1$, there exists $m\in D$ such that $m\prec m'$, and so $e_m=n\circ h$ for some $n\in\Path{Y}$ satisfying $(m, n)\in\W(X,Y)$. Thus, there is $n'\in\Path{Y}$ such that ${n\prec n'}$ and $(m',n')\in\W(X,Y)$. If $h'\colon \dom(m')\to\dom(n')$ is any morphism in $\C$, we let $e_{m'}\coloneqq n'\circ h'$.
\end{enumerate}
The compatibility condition for the cocone $E$ states that, for all $m,m'\in D$ with $m\leq m'$, $e_{m'}$ extends $e_{m}$; an easy inductive argument shows that it suffices to settle the case where $m\prec m'$. If $\Path\colon\C\to\T$ is faithful, then the full subcategory $\Cp$ of $\C$ defined by the paths is a preorder. Just observe that there is at most one tree morphism between any two finite chains. Therefore, any diagram in $\Cp$ is commutative. Suppose that $e_m=n\circ h$ and $e_{m'}=n'\circ h'$ are as in item~\ref{inductive-step} above. Because $m\prec m'$, there exists an embedding $\dom(m)\emb \dom(m')$. Moreover, $n\prec n'$ entails the existence of an embedding $\dom(n)\emb\dom(n')$ making the rightmost triangle below commute.
\[\begin{tikzcd}[column sep=3em]
\dom(m) \arrow[rightarrowtail]{d} \arrow{r}{h} & \dom(n) \arrow[rightarrowtail]{d} \arrow[rightarrowtail]{r}{n} & Y \\
\dom(m') \arrow{r}{h'} & \dom(n') \arrow[rightarrowtail, shift right]{ur}[swap]{n'} &
\end{tikzcd}\]
Since the leftmost rectangle is a diagram in $\Cp$, it must commute. We conclude that $e_{m'}$ extends $e_{m}$.
\end{proof}

\begin{exas}
In the case of the arboreal categories $\RTk(\sg)$, $\RPk(\sg)$ and $\RMk(\sg)$, it follows from their description in Example~\ref{RTex} that the functor $\Path$ is faithful. In that case, the abstract existential positive game $\EG(X,Y)$ coincides, respectively, with the asymmetric (forth only) versions of the $k$-round Ehrenfeucht-\Fraisse, $k$-pebble and $k$-round bisimulation games. These capture, respectively, equivalence in the existential positive fragments of first-order logic with quantifier rank at most $k$, of $k$-variable first-order logic, and of modal logic with modal depth at most $k$ (in the modal case, the existential positive fragment consists of those formulas that do not use $\Box$ nor $\neg$).
\end{exas}

\begin{rem}
Direct inspection of the proof of Proposition~\ref{pr:exist-game} shows that the assumption that $\Path\colon \C\to\T$ is faithful can be replaced with the weaker condition that there is at most one arrow between any two paths of $\C$.
\end{rem}

Finally, we introduce a third type of game which sits in between $\G(X,Y)$ and $\EG(X,Y)$. The \emph{existential} game $\SEG(X,Y)$ has the same winning relation as the game $\G(X,Y)$, but Spoiler always plays in $\Path{X}$ and Duplicator always responds in $\Path{Y}$.

\begin{prop}\label{pr:strong-exist-game}
Let $\C$ be an arboreal category and let $X,Y$ be any two objects of $\C$. Duplicator has a winning strategy in the game $\SEG(X,Y)$ if, and only if, there exists a pathwise embedding $X\to Y$.  
\end{prop}
\begin{proof}
If $f\colon X\to Y$ is a pathwise embedding, then the set in equation~\eqref{eq:mor-winning-st} is a winning strategy for Duplicator in the game $\SEG(X,Y)$.
Just observe that, since $f$ is a pathwise embedding, for all $m\in\Path{X}$ we have $\Path{f}(m)=f\circ m$ and so $\dom(m)=\dom(\Path{f}(m))$.

For the other direction, suppose that Duplicator has a winning strategy in the game $\SEG(X,Y)$. We proceed as in the proof of Proposition~\ref{pr:exist-game} and define inductively a cocone 
\[
E=\{e_m\colon P\to Y \mid m\in D\}
\] 
over the diagram $D$ of paths that embed in $X$. This time, all the morphisms $e_m$ will be embeddings, and so the induced morphism $X\to Y$ will be a pathwise embedding.

\begin{enumerate}[label=(\roman*)]
\item For the base case, suppose that $m=\bot_X\colon P\emb X$ is the root of $\Path{X}$, and let $\bot_Y\colon Q\emb Y$ be the root of $\Path{Y}$. As $(\bot_X,\bot_Y)\in\W(X,Y)$, we have $P\cong Q$. Upon identifying $P$ and $Q$, we set $e_m\coloneqq \bot_Y$.
\item Assume that, for all $m\in D$ of height at most $k$ in $\Path{X}$, $e_m$ has been defined such that $(m, e_m)\in\W(X,Y)$. If $m'\in D$ has height $k+1$, there exists $m\in D$ such that $m\prec m'$, and so $(m, e_m)\in\W(X,Y)$. Thus, there is $n'\in\Path{Y}$ such that ${e_m\prec n'}$ and $(m',n')\in\W(X,Y)$. Let $e_{m'}\coloneqq n'$. Upon composing with an appropriate isomorphism, we can assume that $\dom(m')=\dom(e_{m'})$.
\end{enumerate}
To verify the compatibility condition for the cocone $E$, it suffices to show that for all $m,m'\in D$, if $m\prec m'$ in $\Path{X}$ then $e_m\prec e_{m'}$ in $\Path{Y}$. In turn, this follows by the definition of $E$.
\end{proof}

\begin{exas}
What does the existential game $\SEG(X,Y)$ correspond to in concrete cases?
For instance, as shown in the forthcoming work~\cite{ALR23}, in the case of the arboreal categories $\RPk(\sg)$ and $\RTk(\sg)$, the game $\SEG(X,Y)$ coincides with the one-way versions of the $k$-pebble and $k$-round Ehrenfeucht-\Fraisse~games, respectively, in which the winning condition is that the ensuing relation is a partial isomorphism. 
For pebble games, this variant of the game was shown in~\cite[Theorem~4.8]{KV1995} (\cf also \cite[Proposition~6]{RosenPhDthesis1995}) to capture the \emph{existential} fragment of (infinitary) $k$-variable logic. This is the fragment where the universal quantifier does not appear, and negation can only be applied to atomic formulas. 
Similarly, for the $k$-round Ehrenfeucht-\Fraisse~games, it captures the existential fragment of first-order logic with quantifier rank at most $k$. Since a formula of quantifier rank $k$ can fruitfully use at most $k$ variables, the latter statement can be deduced \eg from \cite[Proposition~5]{RosenPhDthesis1995} for finite-round pebble games. 
\end{exas}

We conclude this section with an application of the games introduced above. In Section~\ref{s:arboreal-covers}, we shall see that the following lemma is an abstraction of the (trivial) observation that any two $\sg$-structures that are equivalent in a logic $\LL$ are also equivalent in the existential (positive) fragment of $\LL$, \cf Theorem~\ref{t:logic-fragm-character}. Let us say that two objects $X,Y$ in a category are \emph{homomorphically equivalent} if there exist morphisms $X\to Y$ and $Y\to X$ (these need not satisfy any further property).
\begin{lem}\label{l:bisimilar-hom-equivalent}
Let $X,Y$ be bisimilar objects of an arboreal category. Then there exist pathwise embeddings $X\to Y$ and $Y\to X$. In particular, $X$ and $Y$ are homomorphically~equivalent.
\end{lem}
\begin{proof}
This follows directly from Lemma~\ref{l:back-and-forth-sys-games} and Propositions~\ref{p:bisimilar-implies-back-and-forth} and~\ref{pr:strong-exist-game}. Just observe that, if Duplicator has a winning strategy in the game $\G(X,Y)$, then they have a winning strategy in the existential game $\SEG(X,Y)$.
\end{proof}
%

\section{Arboreal Covers}\label{s:arboreal-covers}

We now return to the underlying motivation for the axiomatic development in this paper. Arboreal categories have a rich intrinsic process structure, which allows ``dynamic'' notions such as bisimulation and back-and-forth games, and resource notions such as the height of a tree, to be defined. A key idea is to relate these process notions to extensional, or ``static'' structures. In particular, much of finite model theory and descriptive complexity can be seen in this way.

In the general setting, we have an arboreal category $\C$, and another category $\E$, which we think of as the extensional category. 
\begin{defi}\label{arbcovdef}
An \emph{arboreal cover} of $\E$ by the arboreal category $\C$ is given by a comonadic adjunction 
\[ \begin{tikzcd}
\C \arrow[r, bend left=25, ""{name=U, below}, "L"{above}]
\arrow[r, leftarrow, bend right=25, ""{name=D}, "R"{below}]
& \E.
\arrow[phantom, "\textnormal{\footnotesize{$\bot$}}", from=U, to=D] 
\end{tikzcd}
\]
\end{defi}
As for any adjunction, this induces a comonad on $\E$. The comonad is $(G, \ve, \delta)$, where $G \coloneqq LR$, $\ve$ is the counit of the adjunction, and $\delta_a\colon LRa \to LRLRa$ is given by $\delta_a \coloneqq L(\eta_{Ra})$, with $\eta$ the unit of the adjunction.
The comonadicity condition states that the comparison functor from $\C$ to the Eilenberg-Moore category of coalgebras for this comonad is an isomorphism.
The idea is then that we can use the arboreal category $\C$, with its rich process structure and all the associated notions, to study the extensional category $\E$ via the adjunction. Both the Kleisli category of the comonad, and the full Eilenberg-Moore category, are useful in this regard.

We now bring resources into the picture.
\begin{defi}\label{resarbdef}
Let $\C$ be an arboreal category, with full subcategory of paths $\Cp$. We say that $\C$ is \emph{resource-indexed} by a resource parameter $k$ if for all $k > 0$, there is a full subcategory $\Cp^k$ of $\Cp$ closed under embeddings\footnote{\label{fn:closure-emb}That is, for any embedding $P\emb Q$ in $\C$ with $P,Q$ paths, if $Q\in \Cp^k$ then also $P\in \Cp^k$. We shall further assume that each category $\Cp^k$ contains the initial object of $\C$ (which is always a path).} with
\[ \Cp^1 \into \Cp^2 \into \Cp^3 \into \cdots \]
This induces a corresponding tower of full subcategories $\C_k$ of $\C$, with the objects of $\C_k$ those whose cocone of path embeddings with domain in $\Cp^k$ is a colimit cocone in $\C$.
\end{defi}

\begin{exas}
One resource parameter which is always available is to take $\Cp^k$ to be given by those paths in $\C$ whose chain of subobjects has cardinality $\leq k$. In the case of $\F$ and $\T$, the corresponding categories $\F_k$ and $\T_k$ consist of the forests and trees of height at most $k$, respectively. We can think of this as a temporal parameter, restricting the number of sequential steps, or the number of rounds in a game. For the Ehrenfeucht-\Fraisse~and modal comonads, we recover $\RTk(\sg)$ and $\RMk(\sg)$ as described in Example~\ref{RTex}, corresponding to $k$-round versions of the Ehrenfeucht-\Fraisse~and modal bisimulation games respectively \cite{AS2021}. However, note that for the pebbling comonad, the relevant resource index is the number of pebbles, which is a memory restriction along a computation or play of a game. This leads to $\RPk(\sg)$ as described in Example~\ref{RTex}.
\end{exas}

\begin{rem}
The notion of resource-indexed arboreal category in Definition~\ref{resarbdef} is intended to reflect the fact that, in all examples available to date, there is an order on the set of resources which can be identified with the total order of natural numbers. However it is reasonable to expect that, as new examples arise, it will be necessary to generalize this concept by allowing \eg a partial order of resources; this can be done in a straightforward way by adapting Definition~\ref{resarbdef}.
\end{rem}

In Proposition~\ref{p:C_k-arboreal} below we shall see that, given a resource-indexed arboreal category $\{ \C_k \}$, each category $\C_k$ is arboreal.
This allows us to exploit the ideas developed in this paper for any choice of the resource parameter~$k$. We start by establishing the following facts:
\begin{lem}\label{l:path-Ckp}
Let $\{ \C_k \}$ be a resource-indexed arboreal category. For any path embedding $P\emb Y$ in $\C$, if $Y\in\C_k$ then $P\in \Cp^k$.
\end{lem}
\begin{proof}
Fix an arbitrary path embedding $P\emb Y$ in $\C$. If $Y$ is initial, then the latter embedding is an isomorphism and thus $P$ belongs to $\Cp^k$ by assumption (see Footnote~\ref{fn:closure-emb}). Suppose now that $Y$ is not initial and consider the set
\[
\U\coloneqq \{p\in \Path{Y}\mid \dom(p)\in \Cp^k\}.
\]
Note that $\U$ is well-defined because any two representatives in the equivalence class of $p$ have isomorphic domains, and $\Cp^k$ is closed under isomorphisms. As $Y$ is the colimit in $\C$ of the subdiagram of $\Path{Y}$ consisting of those path embeddings whose domain is in $\Cp^k$, it follows that $\bigvee \U=\id_Y$ in $\Emb{Y}$. Because $Y$ is not initial, the set $\U$ is non-empty. If there exists a path embedding $m\colon P\emb Y$, then $m\leq \bigvee\U$ in $\Path{Y}$ and so, by Proposition~\ref{pr:perfect-lattice-of-strong-subs}, $m$ factors through some $p\in \U$. In particular, there exists an embedding $P\emb \dom(p)$. Because $\dom(p)\in \Cp^k$ and the latter is closed under embeddings, we see that $P\in\Cp^k$.
\end{proof}

\begin{lem}\label{l:C_k-subobjects}
Let $\{ \C_k \}$ be a resource-indexed arboreal category and suppose that $X\emb Y$ is an embedding in $\C$. For any $k$, if $Y\in \C_k$ then also $X\in \C_k$.
\end{lem}
\begin{proof}
Suppose that $j\colon X\emb Y$ is an embedding in $\C$ and $Y\in\C_k$. 
Since $X$ is path-generated, it is the colimit in $\C$ of the small diagram $\Path{X}$. We show that, for any path embedding $P\emb X$, the path $P$ must belong to $\Cp^k$. It then follows that $X\in\C_k$. Let $m\colon P\emb X$ be an arbitrary path embedding. The composite $j\circ m\colon P\emb Y$ is also a path embedding, and so $P\in \Cp^k$ by Lemma~\ref{l:path-Ckp}.
\end{proof}
\begin{prop}\label{p:C_k-arboreal}
Let $\{ \C_k \}$ be a resource-indexed arboreal category. Then $\C_k$ is an arboreal category for each $k$.
\end{prop}
\begin{proof}
If $\C$ is equipped with the stable proper factorisation system $(\Q,\M)$, consider the classes of morphisms $\Q'\coloneqq \Q\cap \C_k$ and $\M'\coloneqq \M\cap \C_k$. Then $(\Q',\M')$ is a proper factorisation system in $\C_k$. Just observe that, whenever $W\to Y$ is a morphism in $\C_k$ and 
\[
W \epi X \emb Y
\]
is its (quotient, embedding) factorisation in $\C$, then $X\in \C_k$ by Lemma~\ref{l:C_k-subobjects}. Using again Lemma~\ref{l:C_k-subobjects}, along with the fact that embeddings are stable under pullbacks, it follows at once that $(\Q',\M')$ is stable (and, in fact, pullbacks of $\Q'$-morphisms along $\M'$-morphisms are computed in $\C$). 

We claim that, with respect to this factorisation system, the paths in $\C_k$ are precisely the objects of $\Cp^k$. It follows by Lemma~\ref{l:C_k-subobjects} that, for all objects $X$ of $\C_k$, the poset of embeddings of $X$ in $\C_k$ is isomorphic to the poset of embeddings of $X$ in $\C$. Thus, an object of $\C_k$ is path in $\C_k$ if, and only if, it is a path in $\C$. In particular, the objects of $\Cp^k$ are paths in~$\C_k$. Conversely, suppose that $P$ is a path in $\C_k$. Applying Lemma~\ref{l:path-Ckp} to the identity arrow $P\to P$ we see that $P\in \Cp^k$.

Next, we show that $\C_k$ satisfies the conditions in the definition of path category. The category $\C_k$ satisfies the 2-out-of-3 condition because $\C$ does.
Moreover, $\C_k$ has all coproducts of small families of paths, and they are computed in~$\C$. To see this, consider a set of paths $\{P_i\in \Cp^k\mid i\in I\}$ and their coproduct $\coprod_{i\in I}P_i$ in~$\C$. If $I=\emptyset$ then $\coprod_{i\in I}P_i$ is initial in $\C$ and thus belongs to $\Cp^k$, hence also to $\C_k$. Suppose that $I\neq \emptyset$. If $m\colon P\emb \coprod_{i\in I}P_i$ is any path embedding in $\C$ then, because $P$ is connected, $m$ must factor through some coproduct injection. In particular, there exist $i\in I$ and an embedding $P\emb P_i$. Since $P_i\in \Cp^k$, we get $P\in\Cp^k$. As $\coprod_{i\in I}P_i$ is path-generated in $\C$, it follows at once that $\coprod_{i\in I}P_i\in \C_k$. Hence, $\coprod_{i\in I}P_i$ coincides with the coproduct of the family $\{P_i\in \Cp^k\mid i\in I\}$ in $\C_k$. 

This shows that $\C_k$ is a path category. Further, every object of $\C_k$ is path-generated by construction, and paths in $\C_k$ are connected because any path in $\C_k$ is also a path in $\C$ and, as observed above, coproducts of paths in $\C_k$ are computed in $\C$. Therefore, $\C_k$ is an arboreal category.
\end{proof}

\begin{defi}
Let $\{ \C_k \}$ be a resource-indexed arboreal category. We define a \emph{resource-indexed arboreal cover} of $\E$ by $\C$ to be an indexed family of comonadic adjunctions
\[ \begin{tikzcd}
\C_k \arrow[r, bend left=25, ""{name=U, below}, "L_k"{above}]
\arrow[r, leftarrow, bend right=25, ""{name=D}, "R_k"{below}]
& \E
\arrow[phantom, "\textnormal{\footnotesize{$\bot$}}", from=U, to=D] 
\end{tikzcd}
\]
with corresponding comonads $G_k$ on $\E$.
\end{defi}

\begin{exas}\label{indarbcovex}
Our key examples arise by taking the extensional category $\E$ to be $\CS$. For each $k > 0$, there are evident forgetful functors
\[ \LE_k\colon \RTk(\sg) \to \CS, \quad \LP_k\colon \RPk(\sg) \to \CS\]
which forget the forest order, and in the case of $\RPk(\sg)$, also the pebbling function. These functors are both comonadic over $\CS$. The right adjoints build a forest over a structure $\As$ by forming sequences of elements over the universe $A$, suitably labelled and with the $\sg$-relations interpreted so as to satisfy the conditions (E) and (P) respectively. 
In the modal logic case, the extensional category $\E$ is the category $\CSstar$ of \emph{pointed} Kripke structures with morphisms the $\sg$-homomorphisms preserving the distinguished point. There is a forgetful functor
\[\LM_k\colon \RMk(\sg) \to \CSstar\]
sending $(\As, {\leq})\in\RMk(\sg)$ to $(\As,a)$, where $a$ is the unique root of $(\As, {\leq})$. This functor is comonadic and its right adjoint sends a pointed Kripke structure $(\As,a)$ to the tree-ordered structure obtained by unravelling the structure $\As$, starting from $a$, to depth $k$, and with the $\sg$-relations interpreted so as to satisfy the condition (M).

These constructions yield the comonads described concretely in \cite{abramsky2017pebbling,AS2021}. The sequences correspond to plays in the Ehrenfeucht-\Fraisse, pebbling and modal bisimulation games respectively.
\end{exas}

\begin{rem}
In all our examples, the adjunctions $L_k\: {\dashv} \: R_k$ in the definition of a resource-indexed arboreal cover satisfy an additional compatibility property. In fact, if $k\leq l$ and $F_{k,l}\colon \C_k\hookrightarrow \C_l$ denotes the inclusion functor, then $L_k = L_l \circ F_{k,l}$. Further, in our key examples the functors $F_{k,l}$ have right adjoints $G_{k,l}\colon \C_l\to \C_k$ by (the dual of) \emph{Freyd's special adjoint functor theorem} (\cf \eg \cite[\S V.8]{MacLane}), and $R_k = G_{k,l} \circ R_l$. However, as these properties are not necessary to develop the general theory of resource-indexed arboreal covers, we shall not assume them.
\end{rem}

We now show how resource-indexed arboreal covers can be used to define important notions on the extensional category.
For a resource-indexed arboreal cover of $\E$ by $\C$, with adjunctions $L_k\: {\dashv} \: R_k$ and comonads $G_k$, we define four resource-indexed relations on objects of $\E$. 
\begin{defi}\label{def:resource-ind-equivalence-rel}
Consider a resource-indexed arboreal cover of $\E$ by $\C$, and any two objects $a,b$ of $\E$. For all $k > 0$, we define:
\begin{itemize}
    \item $a \eqaCk b$ if there are morphisms $R_k a \to R_k b$ and $R_k b\to R_k a$ in $\C_k$.
    \item $a \eqdCk b$ if there are pathwise embeddings $R_k a \to R_k b$ and $R_k b\to R_k a$ in $\C_k$.
    \item $a \eqbCk b$ if there is a bisimulation between $R_k a$ and $R_k b$ in $\C_k$.
    \item $a \eqcCk b$ if there is an isomorphism $R_k a \cong R_k b$ in $\C_k$.
\end{itemize}
\end{defi}
In view of Lemma~\ref{l:bisimilar-hom-equivalent}, we have the following chain of inclusions between these relations:
\[
\eqaCk \ \ \supseteq  \ \ \eqdCk  \ \ \supseteq  \ \ \eqbCk  \ \ \supseteq  \ \ \eqcCk.
\]

\begin{rem}\label{rem:co-kleisli}
The relations $\eqaCk$ and $\eqcCk$ can equivalently be defined in terms of the Kleisli category of $G_k$. In fact, $a \eqaCk b$ if, and only if, there are Kleisli morphisms $G_k a \to b$ and $G_k b \to a$, and $a \eqcCk b$ precisely when $a$ and $b$ are isomorphic in the Kleisli category of $G_k$. However, the relations $\eqdCk$ and $\eqbCk$ require working in the Eilenberg-Moore~category.
\end{rem}

Note that $\eqaCk$, $\eqdCk$ and $\eqcCk$ are always equivalence relations (for the second one, just observe that a composition of pathwise embeddings is a pathwise embedding). Further, $\eqbCk$ is transitive, thus an equivalence relation, whenever $\E$ has binary products. Just observe that, since right adjoints preserve limits, the product of $R_k a$ and $R_k b$ in $\C_k$ exists and can be identified with $R_k (a\times b)$. Hence, transitivity of $\eqbCk$ follows from Corollary~\ref{cor:bisim-equiv-rel}.

\begin{rem}\label{rem:preorder-ext}
Under the assumptions of Definition~\ref{def:resource-ind-equivalence-rel}, we can define a coarsening of the relation $\eqaCk$, namely the preorder $\rightarrow_k^{\C}$ defined by ${a \rightarrow_k^{\C} b}$ if, and only if, there is a morphism $R_k a \to R_k b$ in $\C_k$ (equivalently, a Kleisli morphism $G_k a \to b$). The equivalence relation $\eqaCk$ then coincides with the symmetrisation of~$\rightarrow_k^{\C}$. A similar reasoning applies to $\eqdCk$.
\end{rem}

The equivalence relations $\eqaCk$ and $\eqdCk$ can be characterised using the (forth only) games $\EG$ and $\SEG$, respectively.
\begin{prop}
Consider a resource-indexed arboreal cover of $\E$ by $\C$. The following statements are equivalent for all objects $a,b$ of $\E$:
\begin{enumerate}[label=(\arabic*)]
\item $a \eqdCk b$.
\item Duplicator has a winning strategy in the games $\SEG(R_k a,R_k b)$ and $\SEG(R_k b,R_k a)$.
\end{enumerate}
Moreover, whenever $\Path\colon\C_k\to\T$ is faithful, the following statements are equivalent:
\begin{enumerate}[label=(\arabic*)]\setcounter{enumi}{2}
\item $a \eqaCk b$.
\item Duplicator has a winning strategy in the games $\EG(R_k a,R_k b)$ and $\EG(R_k b,R_k a)$.
\end{enumerate}
\end{prop}
\begin{proof}
In view of Proposition~\ref{p:C_k-arboreal}, the categories $\C_k$ are arboreal. Therefore, the equivalences in the statement follow by Propositions~\ref{pr:strong-exist-game} and~\ref{pr:exist-game}, respectively.
\end{proof}

Similarly, the relation $\eqbCk$ can be characterised in terms of the back-and-forth game $\G$, or equivalently in terms of back-and-forth systems.
\begin{prop}\label{pr:bisim-arb-cover}
Consider a resource-indexed arboreal cover of $\E$ by $\C$, and assume that $\E$ admits binary products. The following statements are equivalent for all objects $a,b$ of $\E$:
\begin{enumerate}[label=(\arabic*)]
\item $a \eqbCk b$.
\item $R_k a$ and $R_k b$ are back-and-forth equivalent.
\item Duplicator has a winning strategy in the game $\G(R_k a,R_k b)$.
\end{enumerate}
\end{prop}
\begin{proof}
This follows directly from Theorems~\ref{th:bisimilar-iff-strong-back-forth} and~\ref{th:gamesiffbisim}, and Proposition~\ref{p:C_k-arboreal}. 
\end{proof}

What do these notions mean in Example~\ref{indarbcovex}? This can be explained by invoking a number of known results from the literature.
For each of our three types of model comparison game, there are corresponding fragments $\Lk$ of (infinitary) first-order logic \emph{without} equality\footnote{In Section~\ref{s:arb-adj}, we will see how to capture fragments of first-order logic \emph{with} equality in the framework of arboreal categories. This is especially relevant for the fragments $\Lk$ and $\SELk$, while $\ELk$ and $\Lck$ often turn out to have the same expressive power as their equality-free fragments.} \cite{Libkin2004,blackburn2002modal}:
\begin{itemize}
\item For Ehrenfeucht-\Fraisse~games, $\Lk$ is the fragment of quantifier rank at most $k$.
\item For pebble games, $\Lk$ is the $k$-variable fragment.
\item For bisimulation games, $\Lk$ is the modal fragment  with modal depth at most $k$.
\end{itemize}
In each case, we write $\ELk$ for the existential positive fragment of $\Lk$, and $\Lck$ for the extension of $\Lk$ with counting quantifiers \cite{Libkin2004}.
For each logic $\LL$, there is the usual equivalence on $\sg$-structures: $\As \eqL \Bs$ if and only if, for all $\vphi$ in $\LL$, $\As \models \vphi \IFF \Bs \models \vphi$.

We now have the following result from \cite{AS2021}:
\begin{thm}\label{t:logic-fragm-character}
Consider any of the three resource-indexed arboreal covers of $\CS$ by $\C$ in Example~\ref{indarbcovex}.
The following statements hold for all $\sg$-structures $\As$ and $\Bs$:
\begin{enumerate}[label=(\alph*)]
\item $\As \eqELk \Bs$ \tabto{2.3cm} $\IFF$ \tabto{3.5cm} $\As \eqaCk \Bs$.
\item $\As \eqLk \Bs$ \tabto{2.3cm} $\IFF$ \tabto{3.5cm} $\As \eqbCk \Bs$.
\item $\As \eqLck \Bs$ \tabto{2.3cm} $\IFF$ \tabto{3.5cm} $\As \eqcCk \Bs$. \qed
\end{enumerate}
\end{thm}
Note that this is really a family of three theorems, one for each type of game  arising from a resource-indexed arboreal cover $\C$ as in Example~\ref{indarbcovex}. Thus in each case, we capture the salient logical equivalences in syntax-free, categorical form.

\begin{rem}
A similar result holds for the existential fragment $\SELk$ of $\Lk$. In fact, as proved in the forthcoming~\cite{ALR23}, for any of the three resource-indexed arboreal covers of $\CS$ by $\C$ in Example~\ref{indarbcovex} we have $\As \eqSELk \Bs$ $\IFF$ $\As \eqdCk \Bs$.
\end{rem}
We return to the general setting. Given a resource-indexed arboreal cover of $\E$ by $\C$, we know by comonadicity that for each $k$, $\C_k$ is isomorphic to the Eilenberg-Moore category of coalgebras for the comonad $G_k$. For each object $a$ of $\E$, we can ask if it carries a $G_k$-coalgebra structure; that is, whether there is a morphism $\alpha\colon a \to G_k a$ satisfying the $G_k$-coalgebra conditions. Moreover, we can ask for the \emph{least} $k$ such that this is the case. We call this the \emph{coalgebra number} of $a$.

The intuition behind this, as explained in \cite{abramsky2017pebbling,AS2021}, is that the resource parameter is bounding access to the structure, meaning that the lower the value of $k$, the more difficult it is to have a morphism in $\E$ with codomain $G_k a$. So the least $k$ for which this \emph{is} possible is a significant invariant of the structure. This intuition is born out by the following result from \cite{abramsky2017pebbling,AS2021}.
\begin{thm}
The following statements hold:
\begin{enumerate}[label=(\alph*)]\label{thm:introCoalgebra}
\item For the Ehrenfeucht-\Fraisse~comonad, the coalgebra number of $\As$ corresponds precisely to the \emph{tree-depth} of $\As$.
\item For the pebbling comonad, the coalgebra number of $\As$ corresponds precisely to the \emph{tree-width} of $\As$.
\item For the modal comonad, the coalgebra number of $\As$ corresponds precisely to the \emph{modal unfolding depth} of $\As$. \qed
\end{enumerate}
\end{thm}
What underlies these results is the comonadicity of the arboreal covers, which means that the coalgebras are witnesses for the various forms of tree decompositions of structures in $\E$ corresponding to these combinatorial invariants.

\subsection{Arboreal adjunctions}\label{s:arb-adj}
Consider the Ehrenfeucht-\Fraisse~resource-indexed arboreal cover $\{\RTk(\sg)\}$ of $\CS$. In view of Theorem~\ref{t:logic-fragm-character}, the equivalence relation
\[
\leftrightarrow_k^{\RT}
\]
on $\CS$ captures precisely equivalence in first-order logic \emph{without} equality and with quantifier rank at most~$k$. In this section, following~\cite{AS2021}, we show how to capture first-order logic \emph{with} equality and quantifier rank at most $k$, denoted by $\FO_k$. 

Let $\sg^I\coloneqq \sg\cup \{I\}$ be the vocabulary obtained by adding a new binary relation symbol~$I$ to~$\sg$. There is an associated Ehrenfeucht-\Fraisse~resource-indexed arboreal cover of $\CSplus$:
\[ \begin{tikzcd}
{\RTk(\sg^I)} \arrow[r, bend left=25, ""{name=U, below}, "L^I_k"{above}]
\arrow[r, leftarrow, bend right=25, ""{name=D}, "R^I_k"{below}]
& {\CSplus}
\arrow[phantom, "\textnormal{\footnotesize{$\bot$}}", from=U, to=D] 
\end{tikzcd}\]

Any $\sg$-structure can be expanded to a $\sg^I$-structure by interpreting $I$ as the identity relation. This yields a full and faithful functor $J\colon \CS\to \CSplus$. The functor $J$ has a left adjoint $H\colon \CSplus\to \CS$ which sends a $\sg^I$-structure $\As$ to the quotient of the $\sg$-reduct of $\As$ with respect to the equivalence relation generated by $I^{\As}$ (see \cite[Lemma~25]{DJR2021}). We can then compose the adjunction $H\dashv J$ with the resource-indexed arboreal cover above.
\begin{equation}\label{eq:resource-ind-adj-Ek}
\begin{tikzcd}
{\RTk(\sg^I)} \arrow[r, bend left=25, ""{name=U, below}, "L^I_k"{above}]
\arrow[r, leftarrow, bend right=25, ""{name=D}, "R^I_k"{below}]
& {\CSplus} \arrow[r, bend left=25, ""{name=U', below}, "H"{above}]
\arrow[r, leftarrow, bend right=25, ""{name=D'}, "J"{below}] & {\CS}
\arrow[phantom, "\textnormal{\footnotesize{$\bot$}}", from=U, to=D] 
\arrow[phantom, "\textnormal{\footnotesize{$\bot$}}", from=U', to=D'] 
\end{tikzcd}
\end{equation}
The ensuing adjunction $H L^I_k\dashv R^I_k J$ between $\CS$ and ${\RTk(\sg^I)}$ is not comonadic. For this reason, we need to relax the definition of resource-indexed arboreal cover by removing the comonadicity assumption:
\begin{defi}
Let $\{ \C_k \}$ be a resource-indexed arboreal category. A \emph{resource-indexed arboreal adjunction} between $\E$ and $\C$ is an indexed family of adjunctions
\[ \begin{tikzcd}
\C_k \arrow[r, bend left=25, ""{name=U, below}, "L_k"{above}]
\arrow[r, leftarrow, bend right=25, ""{name=D}, "R_k"{below}]
& \E
\arrow[phantom, "\textnormal{\footnotesize{$\bot$}}", from=U, to=D] 
\end{tikzcd}
\]
with corresponding comonads $G_k$ on $\E$.
\end{defi}

Definition~\ref{def:resource-ind-equivalence-rel} (see also Remark~\ref{rem:preorder-ext}), which associates with each resource-indexed arboreal cover four relations on the extensional category, extends to any resource-indexed arboreal adjunction in the obvious manner. 
We have the following easy yet useful observations.
\begin{lem}\label{l:equiv-rel-properties}
The following statements hold for any resource-indexed arboreal adjunction between $\E$ and $\C$, with adjunctions $L_k\: {\dashv} \: R_k$ and comonads $G_k$, and all objects $a,b$ of $\E$:
\begin{enumerate}[label=(\alph*)]
\item\label{k-equiv} $a\eqaCk G_k a$, for all $k > 0$.
\item\label{hom-k-hom} If there exists a morphism $a\to b$, then $a \rightarrow_k^{\C} b$ for all $k > 0$.
\end{enumerate}
\end{lem}

\begin{proof}
For item~\ref{k-equiv}, reasoning in the Kleisli category (\cf Remark~\ref{rem:co-kleisli}), we must exhibit two Kleisli morphisms $G_k a \to G_k G_k a$ and $G_k G_k a \to G_k a$. The comultiplication of the comonad $G_k$ yields an arrow $\delta_a\colon G_k a \to G_k G_k a$, while the image under $G_k$ of the counit $\epsilon_a\colon G_k a \to a$ yields an arrow $G_k G_k a \to G_k a$.

For item~\ref{hom-k-hom}, if there is an arrow $a\to b$ in $\E$ then we can simply take its image under $R_k$ to obtain an arrow $R_k a \to R_k b$. By definition (see Remark~\ref{rem:preorder-ext}), we have $a \rightarrow_k^{\C} b$.
\end{proof}

Concretely, in the setting of Theorem~\ref{t:logic-fragm-character}, item~\ref{k-equiv} in the previous lemma states that \[\As \eqELk G_k\As\] for all $\sg$-structures $\As$, whereas item~\ref{hom-k-hom} states that the existential positive fragments $\ELk$ are preserved under homomorphisms.

For the Ehrenfeucht-\Fraisse~resource-indexed arboreal adjunction between $\CS$ and ${\RTk(\sg^I)}$, see equation~\eqref{eq:resource-ind-adj-Ek}, we obtain equivalence relations on $\CS$ that we shall denote by $\eqak^I$, $\eqdk^I$, $\eqbk^I$ and $\eqck^I$.
The relation $\eqbk^I$ captures precisely equivalence of $\sg$-structures in first-order logic with equality and quantifier rank at most $k$, see~\cite[Theorem~10.5]{AS2021}. That is, for all $\sg$-structures $\As,\Bs$, 
\[
\As \eqbk^I \Bs \ \Longleftrightarrow \ \As \equiv^{\FO_k} \Bs.
\]

\begin{rem}
The paths in the category $\RTk(\sg^I)$ are those finite $\sg^I$-structures carrying a linear order (note that linear orders automatically satisfy condition (E) in Example~\ref{RTex}). In~\cite{AS2021}, the bisimilarity relation $\eqbk^I$ is defined  with respect to a smaller collection of paths, namely those linearly ordered finite $\sg^I$-structures in which the relation~$I$ is interpreted as the identity. Nonetheless, these two choices of paths give rise to the same notion of bisimilarity; more details will appear in the forthcoming~\cite{JR2022}.
\end{rem} 

By similar reasoning, it follows that the relation $\eqdk^I$ captures equivalence in the existential fragment $\exists\FO_k$ of $\FO_k$ (more details will appear in~\cite{ALR23}).

Furthermore, 
\[
\As \eqak^I \Bs \ \Longleftrightarrow \ \As \equiv^{\EFO_k} \Bs
\]
and the equivalence relation $\eqak^I$ coincides with the equivalence relation $\leftrightarrows_k^{\RT}$ on $\CS$ induced by the Ehrenfeucht-\Fraisse~resource-indexed arboreal cover $\{\RTk(\sg)\}$ of $\CS$. 
This is a consequence of a finer correspondence between Kleisli morphisms and Duplicator winning strategies, see \cite[Theorems~3.2 and~5.1]{AS2021}. 

The equality ${\eqak^I}={\leftrightarrows_k^{\RT}}$ can also be understood as an equality elimination property, stating that two $\sg$-structures satisfy the same sentences in $\EFO_k$ if, and only if, they satisfy the same sentences in the fragment of $\EFO_k$ without equality. To see why this holds, just observe that the vocabulary $\sg$ is purely relational, so the only atomic formulas using the equality symbol are those of the form $x=y$ with $x,y$ variables. In turn, we can remove these atoms and substitute one variable for the other in the appropriate way. After finitely many steps we obtain a sentence, equivalent to the one we started with and with at most the same quantifier rank, which does not use the equality symbol.

Reasoning along the same lines, it is possible to define a pebbling resource-indexed arboreal adjunction which captures equivalence in $\FO^k$, $k$-variable first-order logic with equality.

\begin{rem}
What about the equivalence relation $\eqck^I$ on $\CS$ induced by the Ehrenfeucht-\Fraisse~resource-indexed arboreal adjunction in equation~\eqref{eq:resource-ind-adj-Ek}? In view of Theorem~\ref{t:logic-fragm-character}, we expect~$\eqck^I$ to capture $\FO_k$ with counting quantifiers. This is indeed the case, and again we have an equality elimination result: the equivalence relation on $\sg$-structures determined by $\FO_k$ with (infinitary conjunctions and disjunctions, and) counting quantifiers coincides with the one induced by its fragment without equality; likewise for (infinitary) $k$-variable logic $\FO^k$ with counting quantifiers. See \cite[Theorem~32]{DJR2021}.
\end{rem}

\section{Further Directions}

As explained in the Introduction, the work on arboreal categories was motivated by the previous and ongoing work on using game comonads and allied methods to give a structural perspective on finite model theory and descriptive complexity.
The challenge was to give an axiomatisation of these constructions adequate to cover all the current examples, while also allowing scope for further developments.
As the results in this paper show, this has been successfully accomplished.

We mention a number of directions for future work:
\begin{itemize}
\item Firstly, the body of work on the ``concrete'' level continues to develop. Examples include work on composition methods \cite{jakl2023categorical}; for a  recent survey, already somewhat out of date, see \cite{abramsky2022structure}.
Can these results be lifted to the level of arboreal categories? Are specific assumptions needed?
An example of work in this direction is \cite{abramsky2023linear}, which lifts the concrete results from \cite{MS2022} to the level of arboreal categories.
\item More broadly, we can see how far it is possible to prove substantial model-theoretic results at the level of arboreal categories, thus opening the prospect of a resource-sensitive form of axiomatic model theory.
An example of this is \cite{abramsky2022arboreal}, which carries through a proof of a generalized form of Rossman's Equirank Homomorphism Preservation Theorem \cite{Rossman2008} at the level of arboreal categories.
\item Finally, we can look for new examples of arboreal categories, which may be of a different flavour to those considered so far. We can also seek to characterise classes of game comonads in terms of conditions on the arboreal categories which give rise to them.
Some elements of a structure theory of arboreal categories which are locally finitely presentable, and their finitely accessible arboreal adjunctions, are developed in \cite{reggio2023finitely}.
\end{itemize}

\section*{Acknowledgment}
  \noindent We acknowledge useful feedback from Colin Riba, Tom\'a\v{s} Jakl and Dan Marsden, and are grateful to the anonymous referees for their careful reading of the present paper and for the numerous comments that contributed to improving the presentation of our results.

\bibliographystyle{alphaurl}
\bibliography{arboreal-ext}

\begin{thebibliography}{{\noopsort{Conghaile}}D21}

\bibitem[Abr22]{abramsky2022structure}
S.~Abramsky.
\newblock Structure and power: an emerging landscape.
\newblock {\em Fundamenta Informaticae}, 186(1--4):1--26, 2022.

\bibitem[ADW17]{abramsky2017pebbling}
S.~Abramsky, A.~Dawar, and P.~Wang.
\newblock The pebbling comonad in finite model theory.
\newblock In {\em Proceedings of the 36th Annual {ACM/IEEE} Symposium on Logic
  in Computer Science, {LICS}}, 2017.

\bibitem[AHS04]{adamek2004abstract}
J.~Ad{\'a}mek, H.~Herrlich, and G.E. Strecker.
\newblock {\em {Abstract and concrete categories. The joy of cats}}.
\newblock Online edition, 2004.

\bibitem[AJM00]{abramsky2000full}
S.~Abramsky, R.~Jagadeesan, and P.~Malacaria.
\newblock {Full abstraction for PCF}.
\newblock {\em Information and computation}, 163(2):409--470, 2000.

\bibitem[ALR]{ALR23}
S.~Abramsky, T.~Laure, and L.~Reggio.
\newblock Existential and positive games: a comonadic and axiomatic view.
\newblock In preparation.

\bibitem[AM21]{Guarded2021}
S.~Abramsky and D.~Marsden.
\newblock Comonadic semantics for guarded fragments.
\newblock In {\em Proceedings of the 36th Annual ACM/IEEE Symposium on Logic in
  Computer Science}, LICS '21. IEEE Press, 2021.

\bibitem[AM22]{AM2022}
S.~Abramsky and D.~Marsden.
\newblock {Comonadic semantics for hybrid logic}.
\newblock In S.~Szeider, R.~Ganian, and A.~Silva, editors, {\em 47th
  International Symposium on Mathematical Foundations of Computer Science (MFCS
  2022)}, volume 241 of {\em Leibniz International Proceedings in Informatics
  (LIPIcs)}, pages 7:1--7:14, Dagstuhl, Germany, 2022. Schloss Dagstuhl --
  Leibniz-Zentrum f{\"u}r Informatik.

\bibitem[AMS23]{abramsky2023linear}
S.~Abramsky, Y.~Montacute, and N.~Shah.
\newblock Linear arboreal categories.
\newblock Preprint available at \url{https://arxiv.org/abs/2301.10088}, 2023.

\bibitem[AR21]{AR2021icalp}
S.~Abramsky and L.~Reggio.
\newblock Arboreal categories and resources.
\newblock In {\em Proceedings of the 48th International Colloquium on Automata,
  Languages, and Programming, {ICALP}}, volume 198 of {\em Leibniz
  International Proceedings in Informatics}, pages 115:1--115:20. Schloss
  Dagstuhl -- Leibniz-Zentrum f{\"u}r Informatik, 2021.

\bibitem[AR23]{abramsky2022arboreal}
S.~Abramsky and L.~Reggio.
\newblock Arboreal categories and homomorphism preservation theorems.
\newblock Preprint available at \url{https://arxiv.org/abs/2211.15808}, 2023.

\bibitem[AS18]{DBLP:conf/csl/AbramskyS18}
S.~Abramsky and N.~Shah.
\newblock Relating structure and power: Comonadic semantics for computational
  resources.
\newblock In {\em 27th {EACSL} Annual Conference on Computer Science Logic,
  {CSL}}, pages 2:1--2:17, 2018.

\bibitem[AS21]{AS2021}
S.~Abramsky and N.~Shah.
\newblock Relating structure and power: Comonadic semantics for computational
  resources.
\newblock {\em Journal of Logic and Computation}, 31(6):1390--1428, 2021.

\bibitem[BDRV02]{blackburn2002modal}
P.~Blackburn, M.~De~Rijke, and Y.~Venema.
\newblock {\em Modal Logic}, volume~53 of {\em Cambridge Tracts in Theoretical
  Computer Science}.
\newblock Cambridge University Press, 2002.

\bibitem[Bor94]{Borceux1}
F.~Borceux.
\newblock {\em Handbook of Categorical Algebra: Volume 1, Basic Category
  Theory}.
\newblock Encyclopedia of Mathematics and its Applications. Cambridge
  University Press, 1994.

\bibitem[{\noopsort{Conghaile}}D21]{conghaile2021game}
A.~{\noopsort{Conghaile}}{\'{O} Conghaile} and A.~Dawar.
\newblock Game comonads {\&} generalised quantifiers.
\newblock In {\em 29th {EACSL} Annual Conference on Computer Science Logic,
  {CSL}}, volume 183 of {\em LIPIcs}, pages 16:1--16:17. Schloss Dagstuhl -
  Leibniz-Zentrum f{\"{u}}r Informatik, 2021.

\bibitem[DJR21]{DJR2021}
A.~Dawar, T.~Jakl, and L.~Reggio.
\newblock Lov{\'a}sz-type theorems and game comonads.
\newblock In {\em Proceedings of the 36th Annual {ACM/IEEE} Symposium on Logic
  in Computer Science, {LICS}}, 2021.

\bibitem[DP02]{DP2002}
B.A. Davey and H.A. Priestley.
\newblock {\em Introduction to lattices and order}.
\newblock Cambridge University Press, New York, second edition, 2002.

\bibitem[FK72]{freyd1972categories}
P.J. Freyd and G.M. Kelly.
\newblock Categories of continuous functors, {I}.
\newblock {\em Journal of Pure and Applied Algebra}, 2(3):169--191, 1972.

\bibitem[HO00]{hyland2000full}
J.M.E. Hyland and C.-H.L Ong.
\newblock {On full abstraction for PCF: I, II, and III}.
\newblock {\em Information and computation}, 163(2):285--408, 2000.

\bibitem[JMS23]{jakl2023categorical}
T.~Jakl, D.~Marsden, and N.~Shah.
\newblock A categorical account of composition methods in logic.
\newblock Preprint available at \url{https://arxiv.org/abs/2304.10196}, 2023.

\bibitem[JNW93]{JNW1993}
A.~Joyal, M.~Nielsen, and G.~Winskel.
\newblock Bisimulation and open maps.
\newblock In {\em Proceedings of 8th Annual IEEE Symposium on Logic in Computer
  Science}, pages 418--427, 1993.

\bibitem[JNW96]{JNW1996}
A.~Joyal, M.~Nielsen, and G.~Winskel.
\newblock Bisimulation from open maps.
\newblock {\em Information and Computation}, 127(2):164--185, 1996.

\bibitem[JR22]{JR2022}
T.~Jakl and L.~Reggio.
\newblock An invitation to game comonads.
\newblock Lecture notes for a course at the 33rd European Summer School in
  Logic, Language and Information ({ESSLLI}), in preparation, 2022.

\bibitem[Klo94]{kloks1994treewidth}
T.~Kloks.
\newblock {\em Treewidth: computations and approximations}, volume 842 of {\em
  Springer Lecture Notes in Computer Science}.
\newblock Springer Science \& Business Media, 1994.

\bibitem[KV95]{KV1995}
P.G. Kolaitis and M.Y. Vardi.
\newblock On the expressive power of datalog: Tools and a case study.
\newblock {\em Journal of Computer and System Sciences}, 51(1):110--134, 1995.

\bibitem[Lib04]{Libkin2004}
L.~Libkin.
\newblock {\em Elements of finite model theory}.
\newblock Texts in Theoretical Computer Science. An EATCS Series.
  Springer-Verlag, Berlin, 2004.

\bibitem[ML98]{MacLane}
S.~Mac~Lane.
\newblock {\em Categories for the working mathematician}, volume~5 of {\em
  Graduate Texts in Mathematics}.
\newblock Springer-Verlag, New York, second edition, 1998.

\bibitem[MLM94]{MM1994}
S.~Mac~Lane and I.~Moerdijk.
\newblock {\em Sheaves in geometry and logic}.
\newblock Universitext. Springer-Verlag, New York, 1994.
\newblock A first introduction to topos theory, Corrected reprint of the 1992
  edition.

\bibitem[MS22]{MS2022}
Y.~Montacute and N.~Shah.
\newblock The pebble-relation comonad in finite model theory.
\newblock In {\em Proceedings of the 37th Annual ACM/IEEE Symposium on Logic in
  Computer Science}, LICS '22. IEEE Press, 2022.

\bibitem[NOdM06]{nevsetvril2006tree}
J.~Ne{\v{s}}et{\v{r}}il and P.~Ossona~de Mendez.
\newblock Tree-depth, subgraph coloring and homomorphism bounds.
\newblock {\em European Journal of Combinatorics}, 27(6):1022--1041, 2006.

\bibitem[Ran52]{Raney1952}
G.N. Raney.
\newblock Completely distributive complete lattices.
\newblock {\em Proc. Amer. Math. Soc.}, 3:677--680, 1952.

\bibitem[Rie]{riehl2008factorization}
E.~Riehl.
\newblock Factorization systems.
\newblock Notes available at
  \url{https://emilyriehl.github.io/files/factorization.pdf}.

\bibitem[Ros95]{RosenPhDthesis1995}
E.~Rosen.
\newblock {\em Finite Model Theory and Finite Variable Logics}.
\newblock PhD thesis, University of Pennsylvania, 1995.

\bibitem[Ros08]{Rossman2008}
B.~Rossman.
\newblock Homomorphism preservation theorems.
\newblock {\em J. ACM}, 55(3):15:1--15:53, 2008.

\bibitem[RR23]{reggio2023finitely}
L.~Reggio and C.~Riba.
\newblock Finitely accessible arboreal adjunctions and {H}intikka formulae.
\newblock Preprint available at \url{https://arxiv.org/abs/2304.12709}, 2023.

\bibitem[SIG17]{Church2017}
ACM SIGLOG.
\newblock {Alonzo Church Award}.
\newblock \url{https://siglog.org/winners-of-the-2017-alonzo-church-award/},
  2017.

\end{thebibliography}

\end{document}